\documentclass[a4paper,UKenglish,cleveref, autoref, thm-restate]{lipics-v2021}

\usepackage{todonotes}
\usepackage{tikz,url}
\usepackage{tikz-cd}
\usetikzlibrary{arrows,automata,decorations,patterns}
\usepackage{amsmath,amssymb,amsthm}
\usepackage{epsfig}
\usepackage{mathrsfs}
\usepackage{mathpartir}
\usepackage{tikzit}

\tikzstyle{pinkvertex}=[fill={magenta!40}, draw=black, shape=circle, thick, inner sep=1pt, minimum size=14pt]
\tikzstyle{bluevertex}=[fill={cyan!60}, draw=black, shape=circle, thick, inner sep=1pt, minimum size=14pt]
\tikzstyle{yellowvertex}=[fill=yellow, draw=black, shape=circle, thick, inner sep=1pt, minimum size=14pt]
\tikzstyle{empty node}=[fill=white, shape=circle, inner sep=1pt]

\tikzstyle{filled path}=[-, fill={blue!60}, semithick, fill opacity=0.15]
\tikzstyle{very thick edge}=[-, very thick]
\tikzstyle{arrow}=[->, >=stealth, thick]
\tikzstyle{double arrow}=[<->, thick, >=stealth]
\tikzstyle{fill highlight}=[-, fill={orange!70}, fill opacity=0.5]
\tikzstyle{fill red}=[-, fill=red, fill opacity=0.2]
\tikzstyle{crossed section}=[-, pattern=north west lines]

\usepackage{macros}

\title{Distributed Knowledge in Simplicial Models}

\author{\'Eric Goubault}{LIX, CNRS, \'Ecole Polytechnique, Institut Polytechnique de Paris, Paris, France}{eric.goubault@polytechnique.edu}{https://orcid.org/0000-0002-3198-1863}{}

\author{J\'er\'emy Ledent}{Université Paris Cité, CNRS, IRIF, F-75013, Paris, France}{jeremy.ledent@irif.fr}{https://orcid.org/0000-0001-7375-4725}{}

\author{Sergio Rajsbaum}{Instituto de Matem\'aticas, UNAM, CDMX 04510, Mexico }{rajsbaum@im.unam.mx}{https://orcid.org/0000-0002-0009-5287}{}

\authorrunning{\'E. Goubault, J. Ledent and S. Rajsbaum}

\Copyright{\'Eric Goubault, J\'er\'emy Ledent and Sergio Rajsbaum}



\category{} 

\nolinenumbers 

\hideLIPIcs  

\EventEditors{}
\EventNoEds{0}
\EventLongTitle{}
\EventShortTitle{}
\EventAcronym{}
\EventYear{}
\EventDate{}
\EventLocation{}
\EventLogo{}
\SeriesVolume{}
\ArticleNo{}

\begin{document}
\maketitle

\begin{abstract}
The usual semantics of multi-agent epistemic logic is based on Kripke models, defined in terms of binary relations on a set of possible worlds. Recently, there has been a growing interest in using simplicial complexes rather than graphs, as models for multi-agent epistemic logic.

This approach uses agents' \emph{views} as the fundamental object instead of worlds. A set of views by different agents about a world forms a {simplex}, and a set of simplexes defines a simplicial complex, that can serve as a model for multi-agent epistemic logic.
This new approach reveals topological information that is implicit in Kripke models, because the binary indistinguishability relations are more clearly seen as $n$-ary relations in the simplicial complex.

This paper, written for an economics audience, introduces simplicial models to non-experts and connects distributed computing, epistemic logic and topology.
Our focus is on \emph{distributed  knowledge} and its fixed point, \emph{common distributed knowledge}. These concepts arise when considering the knowledge that a group of agents would acquire, if they could communicate their local knowledge perfectly. While common knowledge has been shown to be related to consensus, we illustrate how distributed knowledge is related to a task weaker to consensus, called \emph{majority consensus}.

We describe three models of communication, some well-known (immediate snapshot), and others less studied (related to broadcast and test-and-set). When majority consensus is solvable, we describe the distributed knowledge that is used to solve it. When it is not solvable, we present a \emph{logical obstruction}, a formula that should always be known according to the task specification, but which the players cannot know.
\end{abstract}

\section{Introduction}

Logics for reasoning  about knowledge and belief in multi-agent systems are of particular interest to distributed systems since the early 1980's. The fundamental role of notions such as common knowledge and their relationship with coordination tasks such as consensus has been thoroughly studied~\cite{fagin,Moses2016}.

One of the earliest formalizations of the logic of knowledge was Hintikka's possible-worlds semantics~\cite{hintikka:1962}.
In this framework, an epistemic situation is modeled as a set of possible worlds, together with a binary relation connecting the worlds that are indistinguishable from the point of view of some agent.
Given such a structure, an agent \emph{knows} a fact precisely when this fact is true in all the worlds that the agent considers possible, based on their partial information.
Nowadays, this point of view is formalized as a Kripke model where all relations are equivalence relations, which can be axiomatized by the logic \textbf{S5}~\cite{fagin}.
Since then, epistemic logic has been interpreted on a wide range of models, including rough sets~\cite{BanerjeeK07}, neighbourhood structures~\cite{chellas:1980,pacuit:2017}, subset space models~\cite{DabrowskiMP96}, topological spaces~\cite{parikhetal:2007,aybuke.phd:2017}, and the central topic of this paper, simplicial complexes~\cite{gandalf-journal,ledent:2019}.
Simplicial complexes arise when we move from using worlds as the primary object, to \emph{perspectives} about possible worlds.
That is, we shift from a global approach to a local one: the primary objects are now the agents' points of view about the world.
The notion of possible world can be derived from this, as a set of compatible local views, that is a simplex of the simplicial complex.

\paragraph*{From global states to local states}

In a distributed system, each agent holds a private value called its \emph{local state}, which encodes all the information that this agent has observed during the current computation.
A \emph{global state}, on the other hand, is an abstraction collecting all the compatible local states of all participating agents.
This approach led to a fruitful theory of distributed computing based on simplicial complexes (see~\cite{herlihyetal:2013} for an overview).
Remarkably, it was shown that there are topological invariants that are preserved while the agents communicate with each other, that in turn determine which distributed tasks can be solved, or how fast they can be solved.

The same idea can be applied to multi-agent epistemic logic, via the notion of \emph{simplicial model}~\cite{gandalf-journal}.
Instead of thinking in terms of possible worlds and indistinguishability relations between the worlds, the focus is first and foremost on modeling the local states of the agents.
Then, a global state for~$n$ agents can be decomposed as a set $\{\ell_1, \ldots, \ell_n\}$, consisting of one local state for each agent.
In a simplicial model, each of these local states is modeled separately as a \emph{vertex} of a simplicial complex.
A global state corresponds to a set of $n$ vertices, that is, an $(n-1)$-dimensional \emph{simplex}.

Simplicial models are thus geometric in nature, each world being a higher-dimensional cell, and they can be seen as discrete approximations of topological spaces~\cite{Hatcher}.
Also, the way in which  worlds are connected tells us information about the knowledge of the agents.
For instance, two global states are indistinguishable by some agent when it has the same local state in both of them. In a simplicial model, this fact is represented by a situation where one vertex (the local state of the agent) is shared between two different simplexes (the two global states).
More generally, simplicial complexes encode explicitly higher order indistinguishability relations;   
when two global states are indistinguishable by $k$ agents,  the fact is represented by a situation where $k$ vertices are  shared between two  simplexes.
Indistinguishability plays a central role in fault-tolerant distributed computing, and in other areas of computer science~\cite{Indistinguishability}, but also, data science~\cite{TDA} and in a variety of biological and social systems~\cite{BATTISTON20201}. More and more attention has been devoted to the higher-order architecture of such complex systems, as increasing evidence emerges  of how their higher-order structure determines their dynamical behavior.

Various aspects of simplicial models have been recently developed in a number of papers~\cite{Armenta-SeguraL22,BilkovaDKR24,CachinLS25,CACHINsynergistic2025,Castaneda22pattern,Ditmarsch2020KnowledgeAS,goubaultSemisimplicialSetModels2023,faultAgentsBoletin2024,hans2and3,DitmarschGLLR21,Ditmarsch24deadL}, including distributed knowledge, bisimulations and covering spaces, belief, more general hypergraph and simplicial sets, faulty agents, multivalued semantics, dynamics, synergistic knowledge, approximate agreement. Several workshops  have been devoted to the topic, including Dagstuhl~\cite{dagstuhl-CastanedaDKM023} and CELT associated to LICS 2022, see~\url{https://easychair.org/cfp/CELT2022}.

\paragraph*{A classic example: the muddy children puzzle}

To give an intuitive understanding of how simplicial complexes are used to model epistemic situations, let us revisit the classic epistemic logic muddy children puzzle.

\begin{center}
\begin{minipage}{0.9\textwidth}
\emph{Three children are playing in a schoolyard. Some of them have mud on their foreheads. They can all see the mud on each other's forehead, but not on their own. The teacher announces, ``At least one of you has mud on their forehead.''
Then, she asks, ``Raise your hand if you know that you have mud on your forehead.''
No one raises their hand.
The teacher repeats the same question, but still no one raises their hand.
The third time, however, all three children correctly raise their hands.
How did they manage to reach this conclusion?
}
\end{minipage}
\end{center}

Our goal here is to analyze this puzzle using a simplicial model,  and
contrasting it  to the classic Kripke model approach, depicted in Figure~\ref{fig:muddyc-Kripke}.  
We stress that both models carry the same information, formally as dual categories~\cite{gandalf-journal}, but we will try to convince the reader of the benefits of working with the simplicial version. In distributed computing a similar discussion  has been described~\cite{herlihyetal:2013}.
\begin{figure}[h]
\centering
\includegraphics[scale=0.4]{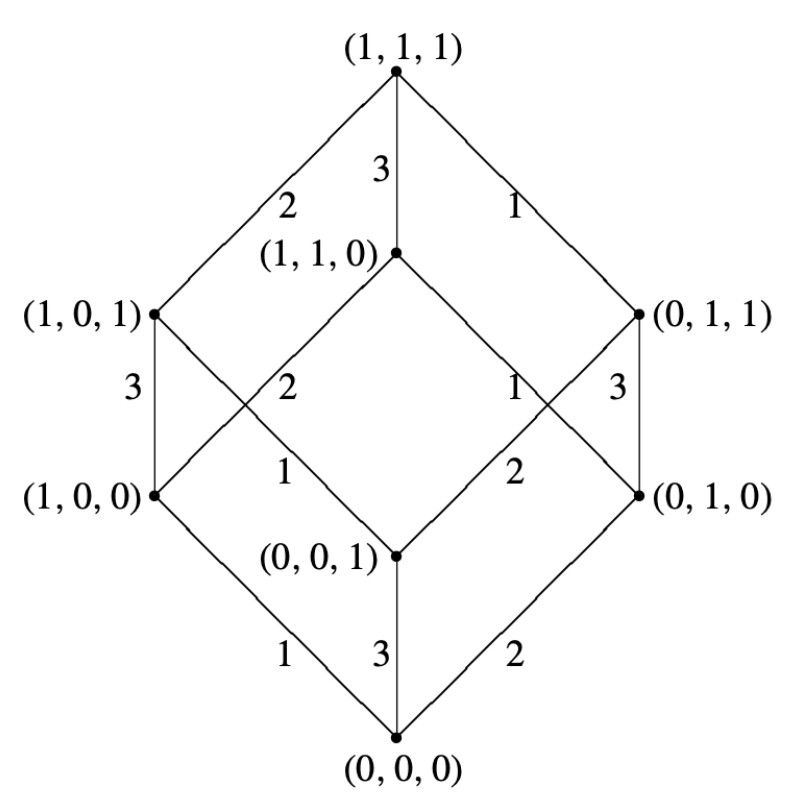}
\caption{The muddy children puzzle:  Kripke model figure from~\cite[Figure 2.3]{halpernmoses:1990}.}
\label{fig:muddyc-Kripke}
\end{figure}

Since we consider $3$ agents (the three children --- we do not model the teacher), the simplicial complex will be of dimension~$2$. That is, each possible world will be represented by a $2$-simplex, i.e., a triangle\footnote{In the whole paper, all of our examples will have only $3$ agents, to make the presentation more accessible, and to be able to illustrate them with pictures. In the general case, to model epistemic situations with~$n$ agents, we need to use $(n-1)$-dimensional complexes}.
In the models depicted in \cref{fig:muddyc}, the three children are represented as colors: pink, blue, and yellow.
In the initial model (left), each child has four vertices: they correspond to the four possible local states of a child.
Indeed, each child can see the forehead of the two other children, giving four possible combinations: $(\text{clean}, \text{clean})$, $(\text{clean}, \text{muddy})$, $(\text{muddy}, \text{clean})$, $(\text{muddy}, \text{muddy})$.
For simplicity, let us use the~`$0$' symbol to mean clean, and `$1$' to mean muddy.
These twelve local states ($4$ for each child) can be assembled into a simplicial model by considering the global states.
There are eight global states, or possible worlds: $000$, $001$, $010$, $011$, $100$, $101$, $110$, $111$.
Each global state indicates the status (clean or muddy) of each child (pink, blue, yellow), from left to right in that order.
A global state is represented as a triangle, linking together three compatible local states.
Notice how the way those triangles are linked together gives information about the agents' uncertainty about the world.
For instance, the worlds $000$ and $100$ share the same pink-colored vertex.
This indicates that, from the local point of view of the pink agent, who can only see that the other two children are clean, both situations are considered possible.

\begin{figure}[h]
\includegraphics[scale=0.2]{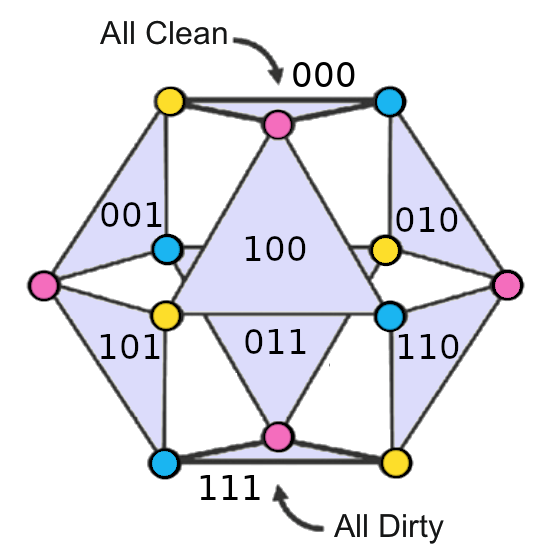}
\hfill
\includegraphics[scale=0.2]{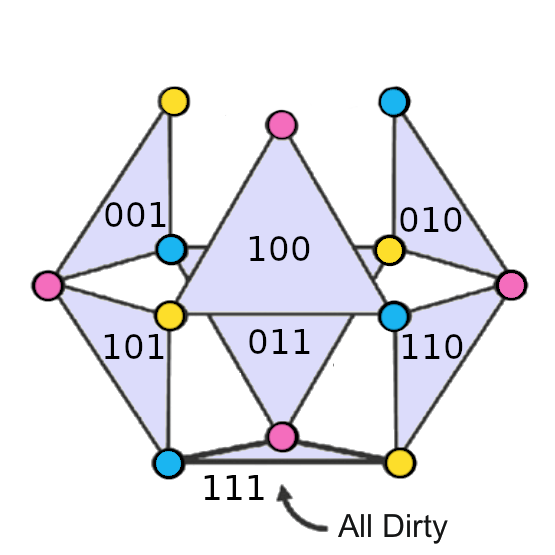}
\hfill
\includegraphics[scale=0.2]{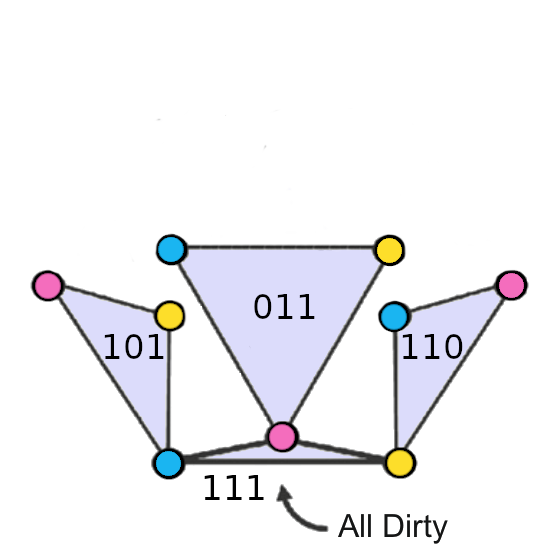}
\caption{The muddy children puzzle: initial situation (left), after the teacher announcement (middle), and after the teacher's first question (right).}
\label{fig:muddyc}
\end{figure}

The model on the left of \cref{fig:muddyc} describes the initial situation of the puzzle.
After the teacher announces that \emph{``At least one of you has mud on their forehead''}, the world $000$ is eliminated. This corresponds to the model depicted in the middle of \cref{fig:muddyc}.
In this new model, there is \emph{common knowledge} that at least one child is muddy.
If the actual world was $100$, then the pink agent would now \emph{know} that she has mud on her forehead, since she can see that the other two children are clean.
Similarly, in the worlds $010$ and $001$, the blue (resp., yellow) agent know that they are muddy.

However, when the teacher asks to \emph{``raise your hand if you know that you are muddy''}, no one raises their hand.
Thus, we learn that the actual world is neither $100$, $010$, nor $001$.
The new model is the one depicted on the right of \cref{fig:muddyc}.
We can iterate the same reasoning when the teacher asks the second question, and further eliminate the worlds where some children would know that they are muddy. After the second question, only the world $111$ where all children are muddy remains.
Therefore, the third time the teacher asks the question, all three children are able to raise their hand.

\paragraph*{Distributed Knowledge}

Various notions of \emph{group knowledge} have been studied in the epistemic logic literature.
For instance, given a set of agents $A$, the formula ``everyone among $A$ knows $\phi$'' says that each individual agent $a \in A$ knows the formula $\phi$.
This is easily expressible in propositional epistemic logic, as a conjunction of formulas of the form ``agent $a$ knows $\phi$'', for each $a \in A$.

This is to be contrasted with the \emph{distributed knowledge} operator, which is strictly more expressive than individual knowledge.
The notion of distributed knowledge has been studied since the 90's in the context of distributed systems~\cite{HalpernM90}.
Intuitively, a group~$A$ of agents has distributed knowledge of a formula~$\phi$, when~$\phi$ can be deduced by aggregating all the local information held by all the agents in~$A$.
In other words, this new logical operator describes the knowledge that could be achieved if a group of agents were able to  communicate perfectly.
Thus, it may be the case that a group of agents have distributed knowledge of some formula~$\phi$, even though no individual agent in that group knows~$\phi$.

Interestingly, recent developments in simplicial complex models for epistemic logic have shown the importance of distributed knowledge, as a higher dimensional version of knowledge.
While individual knowledge is tied to $1$-dimensional connectivity (being able to move from one world to another via shared vertices), distributed knowledge among a group of~$k$ agents is related to the $k$-dimensional connectivity of the model (i.e., moving from one world to another via shared edges, triangles, or in general $(k-1)$-dimensional  cells)\footnote{Not to be confused with the notion of  $k$-connectivity of a topological space, defined in terms of extending maps from a sphere to its interior.}.
One of the first works to notice and make use of the connection between distributed knowledge and topology was the logical obstruction to the solvability of the $k$-set-agreement task by  Nishimura~\cite{Nishimura24}.
But of course the importance of $k$-connectivity has been noticed before, such as in various social and biological settings~\cite{BARCELO200197}.

Just as common knowledge can be defined as the infinite iteration of individual knowledge (everyone knows, and everyone knows that everyone knows, and so on), one can similarly consider an iterated form of distributed knowledge.
This notion, called \emph{common distributed knowledge}, was first axiomatized by Baltag and Smets~\cite{BaltagS20,Baltag21CDK}.
This operator is parametrized by a set of sets of agents.
Intuitively, it requires that every group of agents in that set has distributed knowledge of $\phi$; and that every group has distributed knowledge that every group has distributed knowledge of $\phi$, and so on.
Despite its apparently complex definition, this operator has a very clear geometric interpretation in simplicial models. Take, for example, the set of all groups of two agents, common distributed knowledge of some formula $\phi$ means that $\phi$ is true everywhere in the $2$-connected component of the actual world.

\paragraph*{Contributions}

The main goal of this paper is to present an introduction to simplicial models for an economics audience, which is both self-contained and accessible, avoiding technical proofs and details. We expose the beautiful interactions  connecting distributed
computing, epistemic logic and topology. 
Nevertheless, we do provide new results, although we present them here only in concrete cases, focusing on three agents. They can be generalized to arbitrary number of agents, for a more technical publication. 
The novel contribution is the study of the majority consensus task, in the broadcast and test-and-set models.
Especially, the epistemic treatment, showing the role of distributed knowledge, both to design solutions and to prove when no solution exists.

Impossibility results for solving tasks that are easier than consensus have been explored in~\cite{AttiyaFPR25}, when the obstruction has to do with a
\emph{local articulation point} in the output complex: a vertex whose neighborhood, its link in the topological parlance, is disconnected.
We present a more general, epistemic version of this type of property, through distributed knowledge.


\paragraph*{Organization of the paper}

In Section~\ref{sec:mainExamples} we  present background on simplicial complexes and the duality with Kripke frames. Also,   we recall how to represent dynamics with communication patterns. 
In Section~\ref{sec:dc} we 
 present the notion of a distributed computing models, which defines the way the agents communicate with each other, and examples. Then we present the basic notion of computational problem to be solved in a given model of computation.
In Section~\ref{sec:semantics-DK} the notions of knowledge and distributed knowledge on simplicial complexes are developed.
In Section~\ref{sec:application} all comes together, to study the solvability of majority consensus using distributed knowledge, in  three models of computation.
Section~\ref{sec:concl} concludes the paper.

\section{Background}
\label{sec:mainExamples}

We present background on simplicial complexes, Kripke frames, and their duality relation in the first section. In the second section we present background on communication graphs, the tool we use to evolve a simplicial model after communication.

\subsection{Simplicial complexes and Kripke frames}

The traditional possible worlds semantics of multi-agent modal logics relies on the notion of Kripke frame.
Even though we will be using simplicial complexes in the rest of the paper, let us first recall the formal definition of a Kripke frame, as a point of comparison.
In the following, we fix a finite set of agents, denoted $\Ag$.

\begin{definition}
\label{def:kripkeframe}
\label{def:kripkeframemor}
An epistemic frame $M = \la W, R \ra$ is given by a set of possible worlds~$W$, together with a family of equivalence relations on~$W$, $\sim_a\; \subseteq W \times W$, for each agent $a \in \Ag$.
The relation $\sim_a$ is called the indistinguishability relation for agent~$a$.
\end{definition}

Epistemic frames, augmented with a valuation that specifies which atomic formulas are true in each possible world, serve as the basic semantic structures for multi-agent epistemic logic.
Indeed, the key idea is that an agent $a \in \Ag$ \emph{knows} a formula $\phi$ in some world $w$, precisely when the formula $\phi$ is true in all the worlds $w'$ that are indistinguishable from $w$ from the point of view of $a$ (that is, such that $w' \sim_a w$).
With this interpretation, epistemic frames model the multi-agent epistemic logic $\Sfive$.

Instead of working with epistemic frames, we will rely on \emph{chromatic simplicial complexes}.
Such structures have been used extensively in the field of fault-tolerant distributed computing~\cite{herlihyetal:2013}.
They are in some sense dual to epistemic frames, where we shift from a global perspective to a local one~\cite{gandalf-journal}.

\begin{definition}
\label{def:simplicial-complex}
A chromatic simplicial complex is given by a tuple $\C = \langle V,S,\chi \rangle$, where $V$ is a set of vertices, $S \subseteq \Pow{V}$ is a family of non-empty subsets of $V$ called simplexes, and $\chi : V \to \Ag$ is called the coloring map, such that:
\begin{enumerate}
\item \label{sc:1} for all $v \in V$, $\{v\} \in S$, and
\item \label{sc:2} for all $X \in S$, $Y \subseteq X$ implies $Y \in S$, and
\item \label{sc:3} for all $X \in S$, for all $x,y \in X$, if $x \neq y$ then $\chi(x) \neq \chi(y)$. 
\end{enumerate}
\end{definition}

Let us introduce some vocabulary and explain the three conditions above.
The \emph{dimension} of a simplex $X \in S$ is $\dim(X) = \card{X}-1$.
A simplex of dimension~$0$ is called a \emph{vertex}, of dimension~$1$ is an \emph{edge}, and of dimension~$2$ is a \emph{triangle}.
When two simplexes $X, Y \in S$ are such that $Y \subseteq X$, we say that $Y$ is a \emph{face} of $X$.
\begin{itemize}
\item Condition (\ref{sc:1}) says that a vertex $v \in V$ can be identified with the $0$-dimensional simplex $\{v\} \in S$.
\item Condition (\ref{sc:2}) says that the set of simplexes $S$ is downward-closed with respect to set inclusion.
In other words, if we have a simplex $X \in S$, we must also have all the faces of~$X$.
For instance, if a simplicial complex contains a triangle, it must also contain all three edges of the triangle, as well as the three vertices.
\item Condition (\ref{sc:3}) says that the simplicial complex must be \emph{well-colored}. That is, every simplex must have vertices of distinct colors (agents).
\end{itemize}
Simplexes that are maximal with respect to set inclusion are called \emph{worlds} (or \emph{facets} in topology).
A simplicial complex is \emph{pure} when all worlds are of the same dimension.
In the following, we always consider pure simplicial complexes. However, epistemic logics on impure simplicial complexes have also been studied~\cite{Ditmarsch22complete, Ditmarsch24deadL, faultAgentsBoletin2024}.

\begin{example}\label{ex:basicDuality}
The picture below shows an epistemic frame (left) and the corresponding chromatic simplicial complex (right).
The set of agents is $\Ag = \{a, b, c\}$.
The frame on the left consists of three worlds $W = \{w_1, w_2, w_3\}$, and the indistinguishability relations are indicated by labeled edges: namely, $w_1 \sim_a w_2$, and $w_1 \sim_b w_2$, and $w_2 \sim_c w_3$.

The same epistemic structure is depicted on the right as a chromatic simplicial complex.
This complex consists of $6$ vertices $V = \{a_1, a_2, b_1, b_2, c_1, c_2\}$.
The coloring map $\chi : V \to \Ag$ is indicated by the names of the vertices, for example, $\chi(a_1) = a$.
There are three simplexes of dimension $2$, represented as triangles, corresponding to the three worlds: $w_1 = \{a_1, b_1, c_1\}$, $w_2 = \{a_1, b_1, c_2\}$, and $w_3 = \{a_2, b_2, c_2\}$.
Thus, the set of simplexes $S$ contains those three worlds and all their faces: three triangles, eight edges and six vertices.

Notice that the indistinguishability relations, in the simplicial complex, correspond to intersection of worlds.
For example, when agent~$c$ has the local state $c_2$, she cannot tell whether she lives in the world $w_2$ or in the world $w_3$, since these two worlds share a common $c$-colored vertex.
In other words, the worlds $w_2$ and $w_3$ are indistinguishable from the point of view of~$c$.
Similarly, the worlds $w_1$ and $w_2$ share an $ab$-colored edge: this will later be crucial to interpret distributed knowledge among the group $\{a, b\}$.

\begin{center}
\begin{tikzpicture}[auto,scale=1.2]

\node (p) at (-0.5,0) {$w_1$};
\node (q) at (1,0) {$w_2$};
\node (r) at (2.5,0) {$w_3$};
\path (p) edge[bend left] node[above] {$a$} (q)
      (p) edge[bend right] node[below] {$b$} (q)
      (q) edge node[above] {$c$} (r);
 

\draw[thick, draw=black, fill=blue!60, fill opacity=0.15]
  (5,0) -- (6,-0.577) -- (6,0.577) -- cycle;
\draw[thick, draw=black, fill=blue!60, fill opacity=0.15]
  (6,-0.577) -- (6,0.577) -- (7,0) -- cycle;
\draw[thick, draw=black, fill=blue!60, fill opacity=0.15]
  (7,0) -- (8,-0.577) -- (8,0.577) -- cycle;
\node (p') at (5.65,0) {$w_1$};
\node (q') at (6.35,0) {$w_2$};
\node (r') at (7.65,0) {$w_3$};
\node[yellowvertex] (b1) at (5,0) {$c_1$};
\node[pinkvertex] (g1) at (6,-0.577) {$a_1$};
\node[bluevertex] (w1) at (6,0.577) {$b_1$};
\node[yellowvertex] (b2) at (7,0) {$c_2$};
\node[pinkvertex] (g2) at (8,-0.577) {$a_2$};
\node[bluevertex] (w2) at (8,0.577) {$b_2$};
\end{tikzpicture}
\end{center}
\end{example}

\subsection{Dynamics via communication graphs}

Our goal is to model distributed computing protocols using simplicial models.
We use an abstract formalism based on communication graphs~\cite{BaltagS20,Castaneda22pattern,Kuhn11dynamic}, which is both general and simple, and should therefore be accessible to readers unfamiliar with distributed computing.
We will study three models of computation, the well-known immediate snapshot model, and two more novel, unreliable broadcast,  and test-and-set.
All three of them can be described using the communication graph formalism.

Initially, every agent (or process, as called in distributed computing) starts the computation with a private \emph{input value}, taken from a set of possible values (usually integers).
Then, the agents communicate with each other in a round-based manner, in order to learn new information about their initial inputs.
At each round, the agents send their current local state to some, or all, of the other agents.
Within a round, the set of messages that are successfully delivered is represented by a \emph{communication graph}.
Such a graph indicates how information flows between the agents: an arrow from~$a$ to~$b$ in a communication graph indicates that agent~$a$ successfully sends a message to agent~$b$.

When the agents always broadcast their whole local state (i.e., agent~$a$ sends all the information currently known to~$a$), and always remember everything, we have a   \emph{full-information protocol}. A main application of simplicial complexes is to prove impossibility results of distributed computabilty; for this purpose, it is sufficient to assume a full-information protocol.

After a finite number of rounds of communication, the computation stops and the agents must decide on an output value, based on their final local state.

\begin{definition}[Communication graph]
Let $\Ag$ denote the set of agents.
A communication graph $G \subseteq \Ag \times \Ag$ is a reflexive binary relation on the set of agents.
When $G$ is clear from context, we write $a \to b$ instead of $(a,b) \in G$.
The in-neighbourhood of~$a$ in~$G$ is denoted $\inN{G}{a} = \{ b \in \Ag \mid b \to a \}$. 
\end{definition}

Initially, each agent's local state is its input value.
Say, for example, that we have three agents $a, b, c$ with local states $\ell_a, \ell_b, \ell_c$, respectively.
They send messages to each other according to some communication graph $G$. Note that, since $G$ is reflexive, the agents always manage to ``send their local state to themselves''; in other words, they do not forget their current information.
After communication, each agent has a new local state, collecting all the messages that were received.
For instance, let $G = \{ a \to b,\; a \to c,\; b \to c \}$.
Since agent~$a$ did not receive any message, its new local state is $[\ell_a, \bot, \bot]$, where the symbol $\bot$ indicates that no message was received.
Agent~$b$ received a message from~$a$ only, so its new local state is $[\ell_a, \ell_b, \bot]$.
Lastly, agent~$c$ received all messages, so its new local state is $[\ell_a, \ell_b, \ell_c]$.

Communication graphs express how information is exchanged during a single round of communication.
But how do we decide which graph to use for each round?
This depends on the assumptions we make about the distributed communication model that we have in mind: might some messages be lost, can there be process failures?
In the following, we consider a simple class of models,  called \emph{oblivious adversary}, that has been thoroughly studied e.g.~\cite{WinklerPGSS24}: at each round, one graph is selected from a fixed set of possible communication graphs.

\begin{definition}[Distributed communication model]
\label{def:distCommModel}
Fix a set of agents $\Ag$.
A distributed communication model is a set $M$ of communication graphs on $\Ag$.
\end{definition}

We now have all the tools that we need to model distributed computation.
Fix a distributed communication model $M$.
Initially, we have a simplicial complex describing all the possible configurations assigning input values to agents.
Then, at the first round, a communication graph~$G \in M$ is chosen arbitrarily.
For each input configuration $(\ell^0_a)_{a \in \Ag}$, and each communication graph $G \in P$, we get a set of possible new local states $(\ell^1_a)_{a \in \Ag}$, These new local states form one simplex of the simplicial complex describing the knowledge after one round of communication.

We can then iterate this process for any finite number of rounds, until we reach the final \emph{protocol complex} consisting of all the possible final configurations of the system. Each process must then decide on an output value, based on its local state. Namely, each vertex of the final simplicial complex has an associated output value. These values define a simplicial map to an output complex, for a given distributed task to be solved.
Remarkably, from round to round, the protocol complex preserves topological properties, which vary from model to model, as we shall see in the next examples.
The power of the model to solve tasks, depends on the topological properties it preserves.

\section{Distributed computing}
\label{sec:dc}

We present here three examples of distributed computing models, which define the way the agents communicate with each other. Then we present the basic notion of computational problem to be solved in a given model of computation, distributed task.

\subsection{Examples of distributed communication models}
\label{sec:distrModels}

We now introduce three distributed communication models, in the sense of \cref{def:distCommModel}, and illustrate the simplicial complex model describing the knowledge after one round.
In each case, we use a set of three agents $\Ag = \{a, b, c\}$.
Moreover, we assume that the initial input simplicial complex consists of only two triangles: agents $a$ and $b$ always start the computation with input value~$0$, whereas agent~$c$ can start with input value either~$0$ or~$1$.
This input simplicial complex is depicted below.
\begin{center}
\begin{tikzpicture}[auto,scale=1.2,rotate=30]
\draw[thick, draw=black, fill=blue!60, fill opacity=0.15]
  (5,0) -- (6,-0.577) -- (6,0.577) -- cycle;
\draw[thick, draw=black, fill=blue!60, fill opacity=0.15]
  (6,-0.577) -- (6,0.577) -- (7,0) -- cycle;
\node (p') at (5.65,0) {$w$};
\node (q') at (6.35,0) {$w'$};
\node[yellowvertex] (b1) at (5,0) {$c_0$};
\node[pinkvertex] (g1) at (6,-0.577) {$a_0$};
\node[bluevertex] (w1) at (6,0.577) {$b_0$};
\node[yellowvertex] (b2) at (7,0) {$c_1$};
\end{tikzpicture}
\end{center}

\paragraph*{Unreliable broadcast}

In the \emph{unreliable broadcast} model introduced in~\cite{FPR2026}, at each round, all agents  send their local state to all the other agents.
However, the message sent by an agent is received by either all agents, or by none (due to a transient failure at the round; perhaps of the transmitter of the agent, or of the communication media). The agent does not know if its message was delivered. 
We assume that at each round, at least one agent succeeds in broadcasting its state.
Notice that if an agent does not receive a message from an agent, it knows that no other agent received it.

This model can be described by the following communication graphs (reflexive loops are not depicted).
The graph $G_1$ corresponds to the situation where only~$a$ manages to broadcast its value; in $G_2$, both $a$ and $b$ manage to broadcast their value; and in $G_3$, all three agents manage to broadcast their value.
\begin{mathpar}
	\begin{tikzpicture}[auto, scale=1.2, font={\small},
	-{stealth[length=1mm,width=1mm]}, shorten <=1pt,
	{every loop/.style}={looseness=2, min distance=3mm}]
	\node (G) at (-1,0.5) {$G_1 =$};
	\node (a) at (0,1) {$a$};
	\node (b) at (-0.5,0) {$b$};
	\node (c) at (0.5,0) {$c$};
	\path 
		  (a) edge (b)
		  (a) edge (c);
	\end{tikzpicture}
\and
	\begin{tikzpicture}[auto, scale=1.2, font={\small},
	-{stealth[length=1mm,width=1mm]}, shorten <=1pt,
	{every loop/.style}={looseness=2, min distance=3mm}]
	\node (G) at (-1,0.5) {$G_2 =$};
	\node (a) at (0,1) {$a$};
	\node (b) at (-0.5,0) {$b$};
	\node (c) at (0.5,0) {$c$};
	\path 
		  (a.260) edge (b.60)
		  (b.80) edge (a.240)
		  (b) edge (c)
		  (a) edge (c);
	\end{tikzpicture}
\and
	\begin{tikzpicture}[auto, scale=1.2, font={\small},
	-{stealth[length=1mm,width=1mm]}, shorten <=1pt,
	{every loop/.style}={looseness=2, min distance=3mm}]
	\node (G) at (-1,0.5) {$G_3 =$};
	\node (a) at (0,1) {$a$};
	\node (b) at (-0.5,0) {$b$};
	\node (c) at (0.5,0) {$c$};
	\path 
		  (a.260) edge (b.60)
		  (b.80) edge (a.240)
		  (b.10) edge (c.170)
		  (c.190) edge (b.-10)
		  (a.280) edge (c.120)
		  (c.100) edge (a.300);
	\end{tikzpicture}
\end{mathpar}
The distributed communication model $M_\text{UB}$ actually contains seven graphs: $G_1$, $G_2$, $G_3$, and graphs obtained from those three by permuting the names of the agents.
The picture below shows the simplicial complex consisting of all the possible local states after one round of computation.

\begin{center}
\scalebox{0.9}{
	\begin{tikzpicture}[scale=1.2, auto, font={\small}]
	\begin{pgfonlayer}{nodelayer}
		\node [style=yellowvertex] (0) at (-8, 3.5) {$c_0$};
		\node [style=pinkvertex] (1) at (-6, 3.5) {$a_0$};
		\node [style=bluevertex] (2) at (-7, 5.25) {$b_0$};
		\node [style=yellowvertex] (3) at (-5, 5.25) {$c_1$};
		\node [style=none] (5) at (-7, 4.15) {$w$};
		\node [style=none] (7) at (-6, 4.75) {$w'$};
		\node [style=yellowvertex] (8) at (0, 0) {$c_0$};
		\node [style=pinkvertex] (9) at (4, 0) {$a_0$};
		\node [style=yellowvertex] (10) at (6, 0) {$c_0$};
		\node [style=pinkvertex] (11) at (10, 0) {};
		\node [style=bluevertex] (12) at (2, 3.5) {$b_0$};
		\node [style=yellowvertex] (14) at (3, 5.25) {$c_0$};
		\node [style=bluevertex] (16) at (5, 8.75) {};
		\node [style=pinkvertex] (17) at (7, 5.25) {};
		\node [style=bluevertex] (18) at (8, 3.5) {};
		\node [style=bluevertex] (19) at (5, 1.75) {$b_0$};
		\node [style=yellowvertex] (20) at (6, 3.5) {$c_0$};
		\node [style=pinkvertex] (21) at (4, 3.5) {$a_0$};
		\node [style=yellowvertex] (22) at (15, 8.75) {$c_1$};
		\node [style=pinkvertex] (23) at (13, 5.25) {$a_0$};
		\node [style=yellowvertex] (24) at (12, 3.5) {$c_1$};
		\node [style=pinkvertex] (25) at (10, 0) {$a_0$};
		\node [style=bluevertex] (26) at (11, 8.75) {$b_0$};
		\node [style=yellowvertex] (27) at (9, 8.75) {$c_1$};
		\node [style=bluevertex] (28) at (5, 8.75) {$b_0$};
		\node [style=pinkvertex] (29) at (7, 5.25) {$a_0$};
		\node [style=bluevertex] (30) at (8, 3.5) {$b_0$};
		\node [style=bluevertex] (31) at (11, 5.25) {$b_0$};
		\node [style=yellowvertex] (32) at (9, 5.25) {$c_1$};
		\node [style=pinkvertex] (33) at (10, 7.25) {$a_0$};
		\node [style=none] (35) at (8, 1.25) {$G_1$};
		\node [style=none] (37) at (7, 4.15) {$G_2$};
		\node [style=none] (39) at (5, 2.85) {$G_3$};
		\node [style=none] (41) at (10, 5.9) {$G_3$};
		\node [style=none] (43) at (10, 2.25) {$G_1$};
		\node [style=none] (45) at (8, 4.6) {$G_2$};
		\node [style=none] (46) at (-3.5, 4.5) {};
		\node [style=none] (47) at (-0.5, 4.5) {};
		\node [style=none] (48) at (-2, 5) {};
		\node [style=none] (49) at (-2, 5) {One Round};
	\end{pgfonlayer}
	\begin{pgfonlayer}{edgelayer}
		\draw [style=filled path] (3.center)
			 to (1.center)
			 to (0.center)
			 to (2.center)
			 to cycle;
		\draw [style=filled path] (2) to (1);
		\draw [style=filled path] (8.center)
			 to (12.center)
			 to (9.center)
			 to cycle;
		\draw [style=filled path] (19.center)
			 to (10.center)
			 to (9.center)
			 to cycle;
		\draw [style=filled path] (19.center)
			 to (21.center)
			 to (20.center)
			 to cycle;
		\draw [style=filled path] (12.center)
			 to (21.center)
			 to (14.center)
			 to cycle;
		\draw [style=filled path] (17.center)
			 to (18.center)
			 to (20.center)
			 to cycle;
		\draw [style=filled path] (11.center)
			 to (10.center)
			 to (18.center)
			 to cycle;
		\draw [style=filled path] (17.center)
			 to (16.center)
			 to (14.center)
			 to cycle;
		\draw [style=filled path] (22.center)
			 to (26.center)
			 to (23.center)
			 to cycle;
		\draw [style=filled path] (31.center)
			 to (24.center)
			 to (23.center)
			 to cycle;
		\draw [style=filled path] (31.center)
			 to (33.center)
			 to (32.center)
			 to cycle;
		\draw [style=filled path] (26.center)
			 to (33.center)
			 to (27.center)
			 to cycle;
		\draw [style=filled path] (29.center)
			 to (30.center)
			 to (32.center)
			 to cycle;
		\draw [style=filled path] (25.center)
			 to (24.center)
			 to (30.center)
			 to cycle;
		\draw [style=filled path] (29.center)
			 to (28.center)
			 to (27.center)
			 to cycle;
		\draw [style=arrow] (46.center) to (47.center);
		\draw [style=arrow] (8) to (12);
		\draw [style=arrow] (8) to (9);
		\draw [style=arrow] (9) to (19);
		\draw [style=arrow] (10) to (19);
		\draw [style=arrow] (25) to (10);
		\draw [style=arrow] (25) to (30);
		\draw [style=arrow] (25) to (24);
		\draw [style=arrow] (24) to (31);
		\draw [style=arrow] (23) to (31);
		\draw [style=arrow] (30) to (20);
		\draw [style=arrow] (29) to (20);
		\draw [style=arrow] (28) to (14);
		\draw [style=arrow] (28) to (29);
		\draw [style=arrow] (28) to (27);
		\draw [style=arrow] (27) to (33);
		\draw [style=arrow] (26) to (33);
		\draw [style=arrow] (22) to (23);
		\draw [style=arrow] (22) to (26);
		\draw [style=double arrow] (27) to (26);
		\draw [style=double arrow] (23) to (24);
		\draw [style=double arrow] (29) to (30);
		\draw [style=double arrow] (9) to (10);
		\draw [style=double arrow] (14) to (12);
		\draw [style=double arrow] (21) to (20);
		\draw [style=double arrow] (20) to (19);
		\draw [style=double arrow] (19) to (21);
		\draw [style=double arrow] (32) to (33);
		\draw [style=double arrow] (33) to (31);
		\draw [style=double arrow] (31) to (32);
		\draw [style=arrow] (29) to (32);
		\draw [style=arrow] (30) to (32);
	\end{pgfonlayer}
\end{tikzpicture}
}
\end{center}


Notice that the overall structure of the initial simplicial complex, with two possible worlds~$w$ and $w'$, is preserved.
However, both of those worlds have been subdivided into~$7$ triangles, depending on which of the~$7$ communication graphs $G \in M_\text{UB}$ occurred. And in each of the two subdivisions, three holes have appeared (while  the initial simplicial complex did not have any holes).
To obtain the complex after one more round, the communication model is applied once more, for each round one world.
Although more holes will appear, there are topological invariants preserved, such as the $1$-connectivity of the protocol complex.

The triangles corresponding to the graphs $G_1, G_2, G_3$ are labeled accordingly; note that there are two of each, depending on whether we started in world~$w$ or in world~$w'$.
In the figure, the arrows on the edges of the simplicial complex are not part of the simplicial complex structure.
They are here to indicate which communication graph corresponds to each simplex. For instance, in the two simplices labeled with $G_1$, there are arrows $a \to b$ and $a \to c$, indicating that agent~$a$ successfully broadcast its value to~$b$ and~$c$.
Lastly, note that the information within a vertex is not the full local state of the agent: it only contains the name of the agent, and its original input value.
The local state is left implicit, and can be deduced from the set of incoming arrows.

\paragraph*{Immediate snapshot}

The second model that we introduce is the \emph{immediate snapshot} model.
It has been studied extensively in distributed computing, because it induces a transformation of simplicial complexes called the \emph{chromatic subdivision}, where no holes  ever appear.
It appears in various distributed computing settings, where agents communicate by a simple shared memory (composed of read/write registers only) or message passing~\cite{herlihyetal:2013}.

Viewed abstractly as a distributed communication model, the immediate snapshot model can be described by the set $M_\text{IS}$ of communication graphs.
It contains the four graphs depicted below, as well as graphs obtained from them by permuting the names of the agents, totaling $13$ graphs.
\begin{mathpar}
	\begin{tikzpicture}[auto, scale=1.2, font={\small},
	-{stealth[length=1mm,width=1mm]}, shorten <=1pt,
	{every loop/.style}={looseness=2, min distance=3mm}]
	\node (G) at (-1,0.5) {$G_1 =$};
	\node (a) at (0,1) {$a$};
	\node (b) at (-0.5,0) {$b$};
	\node (c) at (0.5,0) {$c$};
	\path 
		  (a) edge (b)
		  (b) edge (c)
		  (a) edge (c);
	\end{tikzpicture}
\and
	\begin{tikzpicture}[auto, scale=1.2, font={\small},
	-{stealth[length=1mm,width=1mm]}, shorten <=1pt,
	{every loop/.style}={looseness=2, min distance=3mm}]
	\node (G) at (-1,0.5) {$G_2 =$};
	\node (a) at (0,1) {$a$};
	\node (b) at (-0.5,0) {$b$};
	\node (c) at (0.5,0) {$c$};
	\path 
		  (a.260) edge (b.60)
		  (b.80) edge (a.240)
		  (b) edge (c)
		  (a) edge (c);
	\end{tikzpicture}
\and
	\begin{tikzpicture}[auto, scale=1.2, font={\small},
	-{stealth[length=1mm,width=1mm]}, shorten <=1pt,
	{every loop/.style}={looseness=2, min distance=3mm}]
	\node (G) at (-1,0.5) {$G_3 =$};
	\node (a) at (0,1) {$a$};
	\node (b) at (-0.5,0) {$b$};
	\node (c) at (0.5,0) {$c$};
	\path 
		  (a.260) edge (b.60)
		  (b.80) edge (a.240)
		  (c) edge (b)
		  (c) edge (a);
	\end{tikzpicture}
\and
	\begin{tikzpicture}[auto, scale=1.2, font={\small},
	-{stealth[length=1mm,width=1mm]}, shorten <=1pt,
	{every loop/.style}={looseness=2, min distance=3mm}]
	\node (G) at (-1,0.5) {$G_4 =$};
	\node (a) at (0,1) {$a$};
	\node (b) at (-0.5,0) {$b$};
	\node (c) at (0.5,0) {$c$};
	\path 
		  (a.260) edge (b.60)
		  (b.80) edge (a.240)
		  (b.10) edge (c.170)
		  (c.190) edge (b.-10)
		  (a.280) edge (c.120)
		  (c.100) edge (a.300);
	\end{tikzpicture}
\end{mathpar}

Computationally, this behavior can be achieved in a shared memory model, where each agent first writes its value in a dedicated memory cell, and then reads the values of all other agents.
Since the model is asynchronous, an agent that is ``too fast'' may  read before the other agents have written their value, and receive no information.
For example, in graph $G_1$, agent~$a$ ran first, writing, but reading no value from the other agents from the shared memory, then agent~$b$ came second and saw the value of~$a$, then agent~$c$ came third and read the values written by the two other agents.
In graph $G_2$, agents~$a$ and~$b$ ran at the same time and only saw each other, then agent~$c$ came third and saw both other values.

The picture below shows how the input simplicial complex evolves after one round of immediate snapshot communication.

\begin{center}
\scalebox{0.9}{
	\begin{tikzpicture}[scale=1.2, auto, font={\small}]
	\begin{pgfonlayer}{nodelayer}
		\node [style=yellowvertex] (0) at (-8, 3.5) {$c_0$};
		\node [style=pinkvertex] (1) at (-6, 3.5) {$a_0$};
		\node [style=bluevertex] (2) at (-7, 5.25) {$b_0$};
		\node [style=yellowvertex] (3) at (-5, 5.25) {$c_1$};
		\node [style=none] (5) at (-7, 4.15) {$w$};
		\node [style=none] (7) at (-6, 4.75) {$w'$};
		\node [style=yellowvertex] (8) at (0, 0) {$c_0$};
		\node [style=pinkvertex] (11) at (10, 0) {};
		\node [style=yellowvertex] (22) at (15, 8.75) {$c_1$};
		\node [style=pinkvertex] (25) at (10, 0) {$a_0$};
		\node [style=bluevertex] (28) at (5, 8.75) {$b_0$};
		\node [style=none] (46) at (-3.5, 4.5) {};
		\node [style=none] (47) at (-0.5, 4.5) {};
		\node [style=none] (48) at (-2, 5) {};
		\node [style=none] (49) at (-2, 5) {One Round};
		\node [style=yellowvertex] (50) at (6.75, 0) {$c_0$};
		\node [style=pinkvertex] (51) at (3.25, 0) {$a_0$};
		\node [style=yellowvertex] (52) at (8.25, 8.75) {$c_1$};
		\node [style=bluevertex] (53) at (11.75, 8.75) {$b_0$};
		\node [style=yellowvertex] (54) at (3.25, 5.75) {$c_0$};
		\node [style=bluevertex] (55) at (1.75, 3) {$b_0$};
		\node [style=pinkvertex] (56) at (6.75, 5.75) {$a_0$};
		\node [style=bluevertex] (57) at (8.25, 3) {$b_0$};
		\node [style=bluevertex] (58) at (5, 2) {$b_0$};
		\node [style=pinkvertex] (59) at (4, 3.5) {$a_0$};
		\node [style=yellowvertex] (60) at (6, 3.5) {$c_0$};
		\node [style=yellowvertex] (61) at (9, 5.25) {$c_1$};
		\node [style=pinkvertex] (62) at (10, 7) {$a_0$};
		\node [style=bluevertex] (63) at (11, 5.25) {$b_0$};
		\node [style=yellowvertex] (64) at (11.75, 3) {$c_1$};
		\node [style=pinkvertex] (65) at (13.25, 5.75) {$a_0$};
		\node [style=none] (66) at (5, 3) {$G_4$};
		\node [style=none] (67) at (7, 4.15) {$G_2$};
		\node [style=none] (68) at (3.6, 2.25) {$G_3$};
		\node [style=none] (69) at (8, 2.25) {$G_1$};
		\node [style=none] (70) at (8, 4.75) {$G_2$};
		\node [style=none] (71) at (9, 2.75) {$G_1$};
		\node [style=none] (72) at (10, 5.85) {$G_4$};
		\node [style=none] (73) at (11.5, 6.75) {$G_3$};
	\end{pgfonlayer}
	\begin{pgfonlayer}{edgelayer}
		\draw [style=filled path] (3.center)
			 to (1.center)
			 to (0.center)
			 to (2.center)
			 to cycle;
		\draw [style=filled path] (2) to (1);
		\draw [style=arrow] (46.center) to (47.center);
		\draw [style=filled path] (8.center)
			 to (28.center)
			 to (22.center)
			 to (25.center)
			 to cycle;
		\draw [style=arrow] (8) to (55);
		\draw [style=arrow] (8) to (51);
		\draw [style=arrow] (8) to (59);
		\draw [style=arrow] (8) to (58);
		\draw [style=arrow] (25) to (50);
		\draw [style=arrow] (25) to (57);
		\draw [style=arrow] (25) to (58);
		\draw [style=arrow] (25) to (60);
		\draw [style=arrow] (25) to (61);
		\draw [style=arrow] (25) to (63);
		\draw [style=arrow] (25) to (64);
		\draw [style=arrow] (22) to (65);
		\draw [style=arrow] (22) to (63);
		\draw [style=arrow] (22) to (53);
		\draw [style=arrow] (22) to (62);
		\draw [style=arrow] (28) to (52);
		\draw [style=arrow] (28) to (62);
		\draw [style=arrow] (28) to (61);
		\draw [style=arrow] (28) to (56);
		\draw [style=arrow] (28) to (60);
		\draw [style=arrow] (28) to (59);
		\draw [style=arrow] (28) to (54);
		\draw [style=arrow] (54) to (59);
		\draw [style=arrow] (55) to (59);
		\draw [style=arrow] (51) to (58);
		\draw [style=arrow] (50) to (58);
		\draw [style=arrow] (56) to (60);
		\draw [style=arrow] (57) to (60);
		\draw [style=arrow] (57) to (61);
		\draw [style=arrow] (56) to (61);
		\draw [style=arrow] (52) to (62);
		\draw [style=arrow] (53) to (62);
		\draw [style=arrow] (65) to (63);
		\draw [style=arrow] (64) to (63);
		\draw [style=double arrow] (65) to (64);
		\draw [style=double arrow] (52) to (53);
		\draw [style=double arrow] (54) to (55);
		\draw [style=double arrow] (51) to (50);
		\draw [style=double arrow] (59) to (60);
		\draw [style=double arrow] (60) to (58);
		\draw [style=double arrow] (58) to (59);
		\draw [style=double arrow] (56) to (57);
		\draw [style=double arrow] (62) to (61);
		\draw [style=double arrow] (61) to (63);
		\draw [style=double arrow] (63) to (62);
	\end{pgfonlayer}
\end{tikzpicture}
}
\end{center}

As in the previous example, we can recognize a pattern where each world of the initial simplicial complex has been subdivided into $13$ triangles, one for each possible communication graph of the model.
The graph corresponding to each triangle can be recognized using the arrows drawn onto the simplicial complex. Moreover, the ones corresponding to the graphs $G_1, G_2, G_3, G_4$ are labeled accordingly.
Within the vertices, the name of the agent and its original input value is indicated.

Notice that the immediate snapshot model  introduces no  holes in the simplicial complex: the topology of the input complex is preserved. And in fact this is the case after any number of rounds.
This is a crucial observation in the distributed computing literature, allowing to prove many impossibility results~\cite{herlihyetal:2013}.

\paragraph*{Test-and-set}

The \emph{test-and-set} object is a synchronization primitive used in concurrent programming to prevent race conditions.
It consists of a memory cell, initialized to~$0$.
Processes can atomically read the current value of the memory cell, and set it to~$1$.
Thus, when several processes try to access it in parallel, a unique process (the ``winner'') will be able to read value~$0$, while all other processes (the ``losers'') will read value~$1$.

Using this primitive, we can use the following variation of the immediate snapshot model. 
At each round, all agents write their value in the shared memory, then they perform a test-and-set. The losers take a snapshot of the shared memory. The winner does not read the memory, but it has the guarantee that all other agents will be able to see its value.

The three communication graphs below correspond to all executions where agent~$a$ is the winner.
Either the other two agents see each other's value (graph $G_3$), or only one sees the other (graphs $G_1$ and $G_2$).
In both cases, they both see the input value of the winner.
The distributed communication model $M_\text{T\&S}$ for the test-and-set model consists of $9$ graphs: the three graphs below, as well as those where agent $b$ or $c$ is the winner.
\begin{mathpar}
	\begin{tikzpicture}[auto, scale=1.2, font={\small},
	-{stealth[length=1mm,width=1mm]}, shorten <=1pt,
	{every loop/.style}={looseness=2, min distance=3mm}]
	\node (G) at (-1,0.5) {$G_1 =$};
	\node (a) at (0,1) {$a$};
	\node (b) at (-0.5,0) {$b$};
	\node (c) at (0.5,0) {$c$};
	\path 
		  (a) edge (b)
		  (b) edge (c)
		  (a) edge (c);
	\end{tikzpicture}
\and
	\begin{tikzpicture}[auto, scale=1.2, font={\small},
	-{stealth[length=1mm,width=1mm]}, shorten <=1pt,
	{every loop/.style}={looseness=2, min distance=3mm}]
	\node (G) at (-1,0.5) {$G_2 =$};
	\node (a) at (0,1) {$a$};
	\node (b) at (-0.5,0) {$b$};
	\node (c) at (0.5,0) {$c$};
	\path 
		  (a) edge (b)
		  (c) edge (b)
		  (a) edge (c);
	\end{tikzpicture}
\and
	\begin{tikzpicture}[auto, scale=1.2, font={\small},
	-{stealth[length=1mm,width=1mm]}, shorten <=1pt,
	{every loop/.style}={looseness=2, min distance=3mm}]
	\node (G) at (-1,0.5) {$G_3 =$};
	\node (a) at (0,1) {$a$};
	\node (b) at (-0.5,0) {$b$};
	\node (c) at (0.5,0) {$c$};
	\path 
		  (a) edge (b)
		  (b.10) edge (c.170)
		  (c.190) edge (b.-10)
		  (a) edge (c);
	\end{tikzpicture}
\end{mathpar}

The simplicial complex describing the knowledge of the agents after one round of the test-and-set protocol is depicted below.
Notice that it is a sub-complex of the immediate snapshot simplicial complex, where some triangles are removed.
Indeed, all communication graphs of $M_\text{T\&S}$ are also in $M_{IS}$. The holes correspond to the four missing graphs.

\begin{center}
\scalebox{0.9}{
	\begin{tikzpicture}[scale=1.2, auto, font={\small}]
	\begin{pgfonlayer}{nodelayer}
		\node [style=yellowvertex] (0) at (-8, 3.5) {$c_0$};
		\node [style=pinkvertex] (1) at (-6, 3.5) {$a_0$};
		\node [style=bluevertex] (2) at (-7, 5.25) {$b_0$};
		\node [style=yellowvertex] (3) at (-5, 5.25) {$c_1$};
		\node [style=none] (5) at (-7, 4.15) {$w$};
		\node [style=none] (7) at (-6, 4.75) {$w'$};
		\node [style=yellowvertex] (8) at (0, 0) {$c_0$};
		\node [style=pinkvertex] (11) at (10, 0) {};
		\node [style=yellowvertex] (22) at (15, 8.75) {$c_1$};
		\node [style=pinkvertex] (25) at (10, 0) {$a_0$};
		\node [style=bluevertex] (28) at (5, 8.75) {$b_0$};
		\node [style=none] (46) at (-3.5, 4.5) {};
		\node [style=none] (47) at (-0.5, 4.5) {};
		\node [style=none] (48) at (-2, 5) {};
		\node [style=none] (49) at (-2, 5) {One Round};
		\node [style=yellowvertex] (50) at (6.25, 0) {$c_0$};
		\node [style=pinkvertex] (51) at (3.75, 0) {$a_0$};
		\node [style=yellowvertex] (52) at (8.75, 8.75) {$c_1$};
		\node [style=bluevertex] (53) at (11.25, 8.75) {$b_0$};
		\node [style=yellowvertex] (54) at (3, 5.25) {$c_0$};
		\node [style=bluevertex] (55) at (2, 3.5) {$b_0$};
		\node [style=pinkvertex] (56) at (7, 5.25) {$a_0$};
		\node [style=bluevertex] (57) at (8, 3.5) {$b_0$};
		\node [style=bluevertex] (58) at (5, 2) {$b_0$};
		\node [style=pinkvertex] (59) at (4, 3.5) {$a_0$};
		\node [style=yellowvertex] (60) at (6, 3.5) {$c_0$};
		\node [style=yellowvertex] (61) at (9, 5.25) {$c_1$};
		\node [style=pinkvertex] (62) at (10, 7) {$a_0$};
		\node [style=bluevertex] (63) at (11, 5.25) {$b_0$};
		\node [style=yellowvertex] (64) at (12, 3.5) {$c_1$};
		\node [style=pinkvertex] (65) at (13, 5.25) {$a_0$};
		\node [style=none] (68) at (6.35, 2.25) {$G_3$};
		\node [style=none] (69) at (8, 2.5) {$G_1$};
		\node [style=none] (71) at (8.9, 3.1) {$G_1$};
		\node [style=none] (73) at (10, 4) {$G_3$};
		\node [style=none] (74) at (6.75, 0.75) {$G_2$};
		\node [style=none] (75) at (11.15, 3.15) {$G_2$};
	\end{pgfonlayer}
	\begin{pgfonlayer}{edgelayer}
		\draw [style=filled path] (3.center)
			 to (1.center)
			 to (0.center)
			 to (2.center)
			 to cycle;
		\draw [style=filled path] (2) to (1);
		\draw [style=arrow] (46.center) to (47.center);
		\draw [style=arrow] (8) to (55);
		\draw [style=arrow] (8) to (51);
		\draw [style=arrow] (8) to (59);
		\draw [style=arrow] (8) to (58);
		\draw [style=arrow] (25) to (50);
		\draw [style=arrow] (25) to (57);
		\draw [style=arrow] (25) to (58);
		\draw [style=arrow] (25) to (60);
		\draw [style=arrow] (25) to (61);
		\draw [style=arrow] (25) to (63);
		\draw [style=arrow] (25) to (64);
		\draw [style=arrow] (22) to (65);
		\draw [style=arrow] (22) to (63);
		\draw [style=arrow] (22) to (53);
		\draw [style=arrow] (22) to (62);
		\draw [style=arrow] (28) to (52);
		\draw [style=arrow] (28) to (62);
		\draw [style=arrow] (28) to (61);
		\draw [style=arrow] (28) to (56);
		\draw [style=arrow] (28) to (60);
		\draw [style=arrow] (28) to (59);
		\draw [style=arrow] (54) to (59);
		\draw [style=arrow] (55) to (59);
		\draw [style=arrow] (51) to (58);
		\draw [style=arrow] (50) to (58);
		\draw [style=arrow] (56) to (60);
		\draw [style=arrow] (57) to (60);
		\draw [style=arrow] (57) to (61);
		\draw [style=arrow] (56) to (61);
		\draw [style=arrow] (52) to (62);
		\draw [style=arrow] (53) to (62);
		\draw [style=arrow] (65) to (63);
		\draw [style=arrow] (64) to (63);
		\draw [style=double arrow] (59) to (60);
		\draw [style=double arrow] (60) to (58);
		\draw [style=double arrow] (58) to (59);
		\draw [style=double arrow] (62) to (61);
		\draw [style=double arrow] (61) to (63);
		\draw [style=double arrow] (63) to (62);
		\draw [style=arrow] (28) to (54);
		\draw [style=filled path] (58.center)
			 to (51.center)
			 to (8.center)
			 to (55.center)
			 to (59.center)
			 to cycle;
		\draw [style=filled path] (25.center)
			 to (64.center)
			 to (63.center)
			 to (61.center)
			 to (57.center)
			 to (60.center)
			 to (58.center)
			 to (50.center)
			 to cycle;
		\draw [style=filled path] (60.center)
			 to (56.center)
			 to (61.center)
			 to (62.center)
			 to (52.center)
			 to (28.center)
			 to (54.center)
			 to (59.center)
			 to cycle;
		\draw [style=filled path] (53.center)
			 to (62.center)
			 to (63.center)
			 to (65.center)
			 to (22.center)
			 to cycle;
	\end{pgfonlayer}
\end{tikzpicture}
}
\end{center}

The reader interested in  simplicial complexes for fours agents, namely, 3-dimensional, can see the figures of~\cite{HerlihyR95} for a version of the test-and-set model.

\subsection{Examples of distributed tasks}
\label{sec:tasks}

The goal of a distributed communication protocol is to solve a \emph{distributed task}.
Each agent starts with an input value, and after several rounds of communication, ends up with a local state encoding some partial information about the system.
Based solely on this local state, the agent must decide on an \emph{output value}.
A task relates input configuration with the corresponding valid output configurations. Solving such a task requires the agents to coordinate their individual decisions so that, together, their outputs satisfy the overall global specification.

Rather than giving the abstract general definition of a task (see~\cite{herlihyetal:2013} for a reference), we focus on two examples: consensus, and majority consensus.

\paragraph*{Consensus}

Fix a set of three agents, $\Ag = \{a, b, c\}$.
Each agent starts with input values within the set $\cB = \{0, 1\}$, so the initial simplicial complex has~$8$ possible input configurations (that is, $8$ triangles), corresponding to all assignments of $0$'s and $1$'s to the agents : $000$, $001$, \ldots, $111$.
The goal of the consensus task is that, at the end of the computation, all agents should agree on the same output value.
Moreover, this common output value should be among the initial input values.
This is traditionally  expressed by the following two properties:
\smallskip
\begin{itemize}
\item \emph{Agreement:} all three output values must be the same.
\item \emph{Validity:} the agreed value must be one of the three inputs.
\end{itemize}
\smallskip
For instance, if the three agents start with inputs $(0,1,0)$, then the three outputs can be either $(1,1,1)$ or $(0,0,0)$.
On the other hand, if the three inputs are $(1,1,1)$, then they must all decide on value~$1$, because agreeing on value~$0$ would contradict the validity condition.

Task specifications also have a geometric interpretation in terms of simplicial complexes.
Indeed, at the end of the computation, agents decide on their output value based on their local state.
But the local states of the agents are precisely the vertices of the simplicial complex representing the knowledge after communication occurs.
So, to solve a task, one must assign a decision value to each vertex of a simplicial complex.
This assignment can also be viewed as a \emph{simplicial map} into another simplicial complex called the \emph{output complex}.

For the consensus task, the output complex contains only two triangles, corresponding to the two possible output configurations: $(0,0,0)$, or $(1,1,1)$.
It is depicted in the figure below, on the right.
The simplicial complex on the left is the input complex, consisting of all the possible input configurations of the task.
Note that, topologically, this is a triangulated sphere.
The task specification $\Delta$ relates simplexes of the input complex with simplexes of the output complex, according to the agreement and validity conditions above.

\begin{center}
    \includegraphics[scale=0.5]{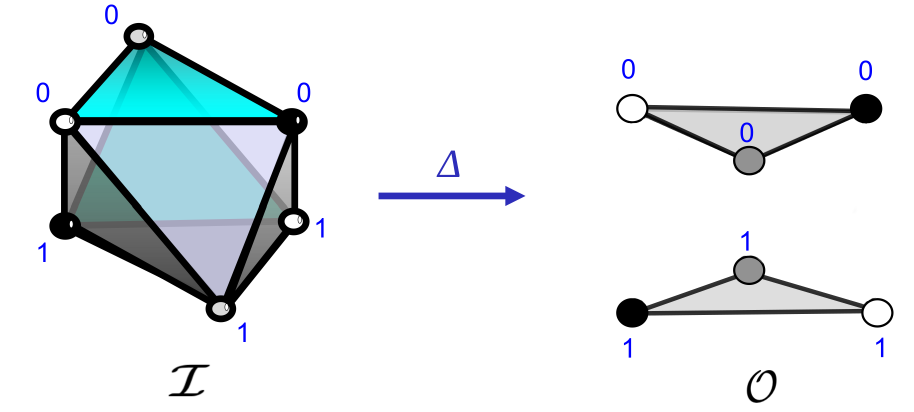}
\end{center}

\paragraph*{Majority consensus}

The consensus task is known to be unsolvable in many distributed computing models~\cite{LynchBook96,facesCons}.
To get around this,  weaker variants of consensus have been defined, such as the \emph{$0$-majority consensus} task~\cite{AttiyaFPR25}.
As in the regular consensus task, agents start with binary input values $0$ or $1$, and decide on binary output values.
If they all start with the same input value, they must all choose this value as output.
However, the agreement condition is relaxed: when the input values are not the same, the agents can all decide $1$, or all decide $0$, or they can disagree as long as there is a majority of $0$'s.
\smallskip
\begin{itemize}
\item \emph{Majority agreement:} either all output values are the same, or there is a majority of $0$'s.
\item \emph{Validity:} all output values must be among the set of inputs.
\end{itemize}

The output complex of the majority consensus task is depicted below. As with consensus, it consists of a subcomplex of the input complex, but
compared to the consensus task, there are three extra triangles, where the output values are $(0,0,1)$, $(0,1,0)$, or $(1,0,0)$.
The task specification $\Delta$ maps the $(0,0,0)$ triangle of the input complex to the $(0,0,0)$ triangle of the output complex; and similarly for the $(1,1,1)$ triangle.
All other triangles of the input complex can be mapped anywhere in the output.

\begin{center}
    \includegraphics[scale=0.5]{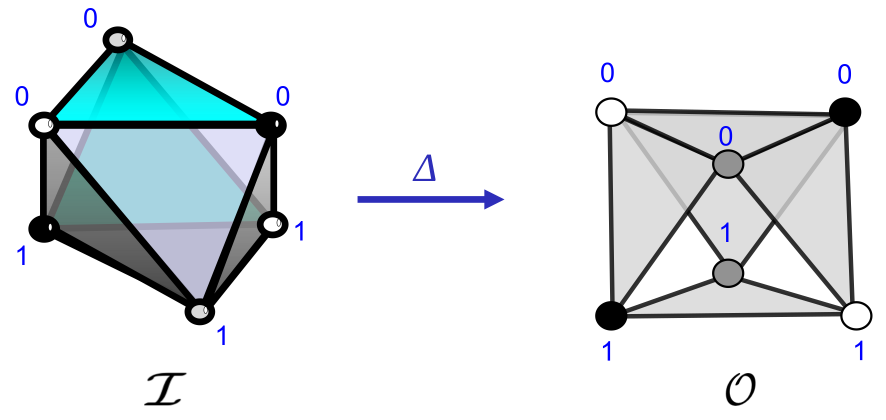}
\end{center}


\section{Interpreting epistemic logics on simplicial complexes}
\label{sec:semantics-DK}

We have seen that simplicial complexes can model a set of possible worlds, and the local points of view of the agents about these worlds.
In order to give a semantics of epistemic logic on such structures, we need an additional piece of information: \emph{atomic formulas}, that is, basic facts about the world that we want to reason about.
In the distributed computing examples, the atomic formulas will be statements of the form ``agent~$a$ has input value~$0$''.
In this section, to illustrate the semantics of the various logical operators, we assume the existence of a single, abstract atomic formula called~$p$.
In the examples, we label the worlds with either $p$ ($p$ is true) or $\neg p$ ($p$ is false).

An epistemic logic formula~$\phi$ is built by combining atomic formulas using various \emph{logical operators}.
First, we have the basic operators of propositional logic:  $\phi$ and $\psi$ (written $\phi \land \psi$); $\phi$ or $\psi$ (written $\phi \lor \psi$); $\phi$ implies $\psi$ (written $\phi \Rightarrow \psi$); not $\phi$ (written $\neg \phi$).
More interesting for our purpose are the epistemic operators:

\begin{itemize}
\item \textbf{Knowledge:} $K_a\,\phi$, where $a \in \Ag$ is an agent. This is read ``$a$ knows $\phi$''.
\item \textbf{Common knowledge:} $C_A\,\phi$, where $A \in \Pow{\Ag}$ is a group of agents. This is read ``there is common knowledge among $A$ that $\phi$''.
\item \textbf{Distributed knowledge:} $D_A\,\phi$, where $A \in \Pow{\Ag}$ is a group of agents. This is read ``there is distributed knowledge among $A$ that $\phi$''.
\item \textbf{Common distributed knowledge:} $CD_\alpha\,\phi$, where $\alpha \in \Pow{\Pow{\Ag}}$ is a set of groups of agents. This is read ``there is common distributed knowledge among $\alpha$ that $\phi$''.
\end{itemize}

In the rest of the section, we explain the intuition behind each of those four operators, and illustrate their simplicial complex semantics using toy examples.

\subsection{Knowledge}

In the classic possible world semantics, the formula $K_a\,\phi$ is true in some world~$w$ when $\phi$ is true in all worlds that are indistinguishable from~$w$ from the point of view of agent~$a$.

Let us translate this in a simplicial complex model.
The actual world~$w$ is a simplex (in our examples, a triangle).
This simplex contains an $a$-colored vertex representing the local state of agent~$a$, i.e., the partial information about the world~$w$ that is available to~$a$.
Any other simplex that contains this vertex is indistinguishable from~$w$, from the point of view of~$a$.
Therefore, we can say that $a$ knows $\phi$ in world $w$, exactly when $\phi$ is true in all simplexes that share an $a$-colored vertex with~$w$.

The simplicial complex below depicts an epistemic situation with $10$ possible worlds. The actual world~$w$ is highlighted in orange.
The atomic proposition $p$ is true in $9$ of the possible worlds (including~$w$), and false in the remaining world at the bottom-left.
\begin{center}
	\begin{tikzpicture}[scale=1.4, auto, font={\small}]
	\begin{pgfonlayer}{nodelayer}
		\node [style=none] (5) at (1, 0.65) {$p$};
		\node [style=yellowvertex] (6) at (1, 1.75) {$c$};
		\node [style=yellowvertex] (7) at (1, -1.75) {$c$};
		\node [style=pinkvertex] (8) at (0, 0) {$a$};
		\node [style=bluevertex] (9) at (2, 0) {$b$};
		\node [style=none] (14) at (1, -0.65) {$p$};
		\node [style=pinkvertex] (17) at (3, 1.75) {$a$};
		\node [style=pinkvertex] (18) at (3, -1.75) {$a$};
		\node [style=yellowvertex] (19) at (4, 0) {$c$};
		\node [style=none] (20) at (3, 0.65) {$p$};
		\node [style=none] (21) at (3, -0.65) {$p$};
		\node [style=none] (22) at (2, -1.1) {$p$};
		\node [style=none] (23) at (2, 1.1) {$p$};
		\node [style=bluevertex] (24) at (-1, 1.75) {$b$};
		\node [style=bluevertex] (25) at (-1, -1.75) {$b$};
		\node [style=yellowvertex] (26) at (-2, 0) {$c$};
		\node [style=none] (27) at (-1, 0.65) {$p$};
		\node [style=none] (28) at (0, -1.1) {$p$};
		\node [style=none] (29) at (0, 1.1) {$p$};
		\node [style=none] (30) at (-1, -0.65) {$\neg p$};
	\end{pgfonlayer}
	\begin{pgfonlayer}{edgelayer}
		\draw [style=filled path] (8) to (7);
		\draw [style=filled path] (24.center)
			 to (26.center)
			 to (25.center)
			 to (7.center)
			 to (18.center)
			 to (19.center)
			 to (17.center)
			 to (6.center)
			 to cycle;
		\draw [style=filled path] (6) to (8);
		\draw [style=filled path] (6) to (9);
		\draw [style=filled path] (9) to (7);
		\draw [style=filled path] (9) to (8);
		\draw [style=filled path] (17) to (9);
		\draw [style=filled path] (9) to (18);
		\draw [style=filled path] (9) to (19);
		\draw [style=filled path] (24) to (8);
		\draw [style=filled path] (8) to (26);
		\draw [style=filled path] (25) to (8);
		\draw [style=fill highlight] (8.center)
			 to (6.center)
			 to (9.center)
			 to cycle;
	\end{pgfonlayer}
\end{tikzpicture}
\end{center}
In this model, in the highlighted world~$w$, it is the case that $b$ knows $p$.
Indeed, the $b$-colored vertex in world $w$ belongs to six possible worlds; but in all of them, $p$ is true.
Thus, the formula $K_b\,\phi$ holds.
Similarly, we can say that $c$ knows $p$ in the world $w$. This time, there are only three possible worlds to consider.
However, agent $a$ does not know $p$. From the point of view of~$a$, the actual world might possibly be the bottom-left one where $p$ is false.
Therefore, the formula $K_a\,p$ is false in world~$w$.

\subsection{Common knowledge}

Common knowledge is the iterated variant of knowledge.
It is parameterized by a set of agents $A \subseteq \Ag$.
Intuitively, the formula $C_A\,\phi$ says that, among the group $A$, everyone knows~$\phi$, and everyone knows that everyone knows~$\phi$, and so on ad infinitum.
Note that common knowledge subsumes individual knowledge, as $K_a\,\phi \Leftrightarrow C_{\{a\}}\,\phi$.

In a simplicial complex model, this can be formalized as follows.
Fix a simplex $w$ representing the actual world.
Then there is common knowledge of~$\phi$ among~$A$ in world~$w$, when $\phi$ is true in all worlds reachable from~$w$ by
a sequence of steps where at each step some agent in $A$ cannot distinguish the two worlds.
A bit more formally, $C_A\,\phi$ holds in world~$w$ when for every finite path $w_0, w_1, w_2, \ldots, w_k$ starting from the world $w_0 = w$, such that at each step, the worlds $w_i$ and $w_{i+1}$ share a vertex whose color belongs to the set $A$, then the formula $\phi$ is true in~$w_k$.

An example is depicted below, with a simplicial complex consisting of six worlds, the actual world~$w$ being highlighted in orange.
\begin{center}
	\begin{tikzpicture}[scale=1.4, auto, font={\small}]
	\begin{pgfonlayer}{nodelayer}
		\node [style=none] (5) at (1, 0.65) {$p$};
		\node [style=yellowvertex] (6) at (1, 1.75) {$c$};
		\node [style=pinkvertex] (8) at (0, 0) {$a$};
		\node [style=bluevertex] (9) at (2, 0) {$b$};
		\node [style=none] (10) at (3, 0.65) {$p$};
		\node [style=yellowvertex] (11) at (4, 0) {$c$};
		\node [style=pinkvertex] (12) at (3, 1.75) {$a$};
		\node [style=bluevertex] (13) at (2, 0) {$b$};
		\node [style=none] (14) at (5, 0.65) {$p$};
		\node [style=yellowvertex] (15) at (4, 0) {$c$};
		\node [style=pinkvertex] (16) at (6, 0) {$a$};
		\node [style=bluevertex] (17) at (5, 1.75) {$b$};
		\node [style=none] (18) at (7, 0.65) {$p$};
		\node [style=yellowvertex] (19) at (7, 1.75) {$c$};
		\node [style=pinkvertex] (20) at (6, 0) {$a$};
		\node [style=bluevertex] (21) at (8, 0) {$b$};
		\node [style=bluevertex] (22) at (8, 0) {$b$};
		\node [style=bluevertex] (23) at (8, 0) {$b$};
		\node [style=none] (24) at (9, 0.65) {$p$};
		\node [style=yellowvertex] (25) at (10, 0) {$c$};
		\node [style=pinkvertex] (26) at (9, 1.75) {$a$};
		\node [style=bluevertex] (27) at (8, 0) {$b$};
		\node [style=none] (28) at (11, 0.65) {$\neg p$};
		\node [style=yellowvertex] (29) at (10, 0) {$c$};
		\node [style=pinkvertex] (30) at (12, 0) {$a$};
		\node [style=bluevertex] (31) at (11, 1.75) {$b$};
		\node [style=pinkvertex] (32) at (12, 0) {$a$};
	\end{pgfonlayer}
	\begin{pgfonlayer}{edgelayer}
		\draw [style=filled path] (8) to (6);
		\draw [style=filled path] (6) to (9);
		\draw [style=filled path] (9) to (8);
		\draw [style=fill highlight] (6.center)
			 to (9.center)
			 to (8.center)
			 to cycle;
		\draw [style=filled path] (12) to (11);
		\draw [style=filled path] (11) to (13);
		\draw [style=filled path] (13) to (12);
		\draw [style=filled path] (11.center)
			 to (13.center)
			 to (12.center)
			 to cycle;
		\draw [style=filled path] (16) to (15);
		\draw [style=filled path] (15) to (17);
		\draw [style=filled path] (17) to (16);
		\draw [style=filled path] (15.center)
			 to (17.center)
			 to (16.center)
			 to cycle;
		\draw [style=filled path] (20) to (19);
		\draw [style=filled path] (19) to (21);
		\draw [style=filled path] (21) to (20);
		\draw [style=filled path] (19.center)
			 to (21.center)
			 to (20.center)
			 to cycle;
		\draw [style=filled path] (26) to (25);
		\draw [style=filled path] (25) to (27);
		\draw [style=filled path] (27) to (26);
		\draw [style=filled path] (25.center)
			 to (27.center)
			 to (26.center)
			 to cycle;
		\draw [style=filled path] (30) to (29);
		\draw [style=filled path] (29) to (31);
		\draw [style=filled path] (31) to (30);
		\draw [style=filled path] (29.center)
			 to (31.center)
			 to (30.center)
			 to cycle;
	\end{pgfonlayer}
\end{tikzpicture}
\end{center}
In the highlighted world~$w$, there is no common knowledge of~$p$ among the agents $\{a, b, c\}$.
Indeed, there is a finite path connecting the world $w$ to the rightmost world where $p$ is false.
The meaning of this path is that, in world~$w$, agent~$b$ considers possible, that agent~$c$ considers possible, that agent~$a$ considers possible, that agent~$b$ considers possible, that agent~$c$ considers possible, that~$p$ is false.
Thus, like in Kripke epistemic models, common knowledge is strongly related to the connectedness of the model.

\subsection{Distributed knowledge}

The distributed knowledge operator $D_A\,\phi$ models, intuitively, what a group~$A$ of agents would know if they were able to combine their individual knowledge (for example, via perfectly reliable communication)~\cite{fagin,halpernmoses:1990}.
Another way to explain it is that we view the group~$A$ of agents as a single entity, which is able to distinguish two possible worlds whenever at least one agent $a \in A$ can distinguish them.
Note that distributed knowledge subsumes individual knowledge, as $K_a\,\phi \Leftrightarrow D_{\{a\}}\,\phi$. However, common knowledge and distributed knowledge are not interdefinable.

In terms of simplicial complex models, there is a clear geometric interpretation of distributed knowledge: it corresponds to higher-dimensional connectivity between adjacent simplexes.
Let~$w$ be the actual world, and $A \subseteq \Ag$ a group of agents.
Then the formula $D_A\,\phi$ holds in world $w$, when $\phi$ is true in all the worlds~$w'$ that share an $A$-colored face with~$w$.
For example, if~$A = \{a, b\}$ is a group of two agents, then distributed knowledge among~$A$ corresponds to moving from one world to another via a shared $ab$-colored edge.

The example below shows how we can have distributed knowledge of a fact, even though each individual agent does not know that fact.
It is a simplicial model with four worlds, the actual world~$w$ being highlighted in orange.
\begin{center}
	\begin{tikzpicture}[scale=1.4, auto, font={\small}]
	\begin{pgfonlayer}{nodelayer}
		\node [style=none] (5) at (1, 0.65) {$p$};
		\node [style=yellowvertex] (6) at (1, 1.75) {$c$};
		\node [style=yellowvertex] (7) at (1, -1.75) {$c$};
		\node [style=pinkvertex] (8) at (0, 0) {$a$};
		\node [style=bluevertex] (9) at (2, 0) {$b$};
		\node [style=yellowvertex] (10) at (3.75, 1) {$c$};
		\node [style=yellowvertex] (11) at (-1.75, -1) {$c$};
		\node [style=bluevertex] (12) at (-1.75, 1) {$b$};
		\node [style=pinkvertex] (13) at (3.75, -1) {$a$};
		\node [style=none] (14) at (1, -0.65) {$p$};
		\node [style=none] (15) at (3.15, 0) {$\neg p$};
		\node [style=none] (16) at (-1.15, 0) {$\neg p$};
	\end{pgfonlayer}
	\begin{pgfonlayer}{edgelayer}
		\draw [style=filled path] (11.center)
			 to (12.center)
			 to (8.center)
			 to (6.center)
			 to (9.center)
			 to (13.center)
			 to (10.center)
			 to (9.center)
			 to (7.center)
			 to (8.center)
			 to cycle;
		\draw [style=filled path] (9) to (8);
		\draw [style=fill highlight] (6.center)
			 to (9.center)
			 to (8.center)
			 to cycle;
	\end{pgfonlayer}
\end{tikzpicture}
\end{center}
First, note that in the highlighted world~$w$, agent~$a$ does not know~$p$. Indeed, $a$ considers possible that the actual world might be the leftmost one, where $p$ is false.
Similarly, agent~$b$ does not know~$p$, because $b$ considers possible that the actual world might be the rightmost one.
However, the formula $D_{\{a,b\}}\,p$ holds in world~$w$: there is distributed knowledge among $\{a, b\}$ that $p$ is true.
This is because there are only two possible worlds that contain the $ab$-colored edge of~$w$, and in both of them, $p$ is true.

\subsection{Common distributed knowledge}
\label{sec:CDK}

Common distributed knowledge is the infinite iteration of distributed knowledge~\cite{BaltagS20,Baltag21CDK}.
It is parameterized by a set of sets of agents, $\alpha \subseteq \Pow{\Ag}$.
Intuitively, the formula $CD_\alpha\,\phi$ says that every group in~$\alpha$ has distributed knowledge of~$\phi$, and that every group has distributed knowledge that every group has distributed knowledge of~$\phi$, and so on ad infinitum.
Note that common distributed knowledge subsumes all the previously introduced operators.
Indeed, we have $K_a\,\phi \Leftrightarrow CD_{\{\{a\}\}}\,\phi$, and $C_A\,\phi \Leftrightarrow CD_{\{\{a\} \,\mid\, a \in A \}}\,\phi$, and $D_A\,\phi \Leftrightarrow CD_{\{A\}}\,\phi$.

In this paper, we will focus on a special case, where $\alpha$ contains all sets of exactly $2$ agents.
That is, $\alpha = \{ \{x, y\} \mid x \neq y \in \Ag \}$.
Since our examples deal with only three agents $a$, $b$, $c$, the resulting set of sets of agents is $\alpha = \{ \{a,b\}, \{a,c\}, \{b,c\}\}$.

For this particular $\alpha$, the semantics of $CD_\alpha\,\phi$ in a simplicial complex can be formulated as follows.
In a world~$w$, here is a common distributed knowledge of~$\phi$ among $\alpha$, when $\phi$ is true in every world~$w'$ in the $2$-connected component of~$w$.
Here, by $2$-connected component, we mean every world reachable from~$w$ by moving to adjacent worlds sharing a common edge.
More formally, $CD_\alpha\,\phi$ holds in world~$w$ when for every finite path $w_0, w_1, w_2, \ldots, w_k$ starting from the world $w_0 = w$, such that at each step, the worlds $w_i$ and $w_{i+1}$ share a common edge, the formula $\phi$ is true in~$w_k$.

The simplicial complex below shows how common distributed knowledge differs from common knowledge.
The actual world~$w$ is highlighted in orange.
\begin{center}
	\begin{tikzpicture}[scale=1.4, auto, font={\small}]
	\begin{pgfonlayer}{nodelayer}
		\node [style=none] (5) at (1, 0.65) {$p$};
		\node [style=yellowvertex] (6) at (1, 1.75) {$c$};
		\node [style=pinkvertex] (8) at (0, 0) {$a$};
		\node [style=bluevertex] (9) at (2, 0) {$b$};
		\node [style=none] (10) at (3, 0.65) {$p$};
		\node [style=none] (14) at (5, 0.65) {$p$};
		\node [style=yellowvertex] (15) at (4, 0) {$c$};
		\node [style=pinkvertex] (16) at (6, 0) {$a$};
		\node [style=bluevertex] (17) at (5, 1.75) {$b$};
		\node [style=none] (18) at (7, 0.65) {$p$};
		\node [style=yellowvertex] (19) at (7, 1.75) {$c$};
		\node [style=bluevertex] (21) at (8, 0) {$b$};
		\node [style=bluevertex] (22) at (8, 0) {$b$};
		\node [style=bluevertex] (23) at (8, 0) {$b$};
		\node [style=none] (24) at (9, 0.65) {$p$};
		\node [style=pinkvertex] (26) at (9, 1.75) {$a$};
		\node [style=none] (28) at (11, 0.65) {$\neg p$};
		\node [style=yellowvertex] (29) at (10, 0) {$c$};
		\node [style=pinkvertex] (30) at (12, 0) {$a$};
		\node [style=bluevertex] (31) at (11, 1.75) {$b$};
		\node [style=pinkvertex] (32) at (12, 0) {$a$};
		\node [style=pinkvertex] (33) at (3, 1.75) {$a$};
		\node [style=none] (34) at (2, 1.15) {$p$};
		\node [style=none] (35) at (8, 1.15) {$p$};
		\node [style=none] (36) at (6, 1.15) {$p$};
		\node [style=none] (37) at (4, 1.15) {$p$};
	\end{pgfonlayer}
	\begin{pgfonlayer}{edgelayer}
		\draw [style=filled path] (8) to (6);
		\draw [style=filled path] (6) to (9);
		\draw [style=filled path] (9) to (8);
		\draw [style=fill highlight] (6.center)
			 to (9.center)
			 to (8.center)
			 to cycle;
		\draw [style=filled path] (16.center)
			 to (15.center)
			 to (9.center)
			 to (6.center)
			 to (33.center)
			 to (17.center)
			 to (19.center)
			 to (26.center)
			 to (29.center)
			 to (23.center)
			 to cycle;
		\draw [style=fill highlight] (16) to (15);
		\draw [style=filled path] (15) to (17);
		\draw [style=filled path] (17) to (16);
		\draw [style=filled path] (15) to (17);
		\draw [style=filled path] (17) to (16);
		\draw [style=filled path] (19) to (21);
		\draw [style=filled path] (19) to (21);
		\draw [style=filled path] (30) to (29);
		\draw [style=filled path] (29) to (31);
		\draw [style=filled path] (31) to (30);
		\draw [style=filled path] (29.center)
			 to (31.center)
			 to (30.center)
			 to cycle;
		\draw [style=filled path] (9) to (33);
		\draw [style=filled path] (33) to (15);
		\draw [style=filled path] (16) to (19);
		\draw [style=filled path] (26) to (23);
	\end{pgfonlayer}
\end{tikzpicture}
\end{center}
Let $\alpha = \{ \{a,b\}, \{a,c\}, \{b,c\}\}$ contain all groups of two agents.
Then, in the highlighted world~$w$, the formula $CD_\alpha\,p$ holds.
Indeed, the $2$-connected component of~$w$ contains all worlds of the model, except for the rightmost one, which is connected only by a single vertex.
So, the world where $p$ is false is not reachable by a path of worlds connected via edges.
In contrast, the formula $C_{\{a,b,c\}}\,p$ does not hold in this model, because standard common knowledge is only concerned with $1$-connectivity.

\section{Studying the solvability of majority consensus}
\label{sec:application}

In this section, our goal is to study the solvability of the majority consensus task in the three distributed computing models introduced in \cref{sec:distrModels}, unreliable broadcast, immediate snapshot and test-and-set.
To show that a task is solvable, it suffices to present an algorithm, by giving the code that each agent executes.
To show that a task is not solvable, however, requires additional techniques.
We describe here an epistemic approach to prove impossibility results.
It relies on finding an epistemic logic formula called a \emph{logical obstruction}~\cite{gandalf-journal} that witnesses the fact that the task cannot be solved.

\subsection{Proving impossibility using an obstruction formula}

Recall that, in the majority consensus task, agents start the computation with a binary input value, $0$ or $1$.
Let $\cI$ denote the \emph{input complex} consisting of all possible binary input assignments.
The \emph{output complex} $\cO$ of the majority consensus task is the one described in \cref{sec:tasks}.
For a given model of communication, we also have a \emph{protocol complex}~$\cP$ describing the local states after $R$ rounds of communication.

In order to solve a task, each agent must decide an output value, based solely on its local state after communicating, in a way that is compatible with the task's requirements of agreement and validity.
Such decision values provide a labeling of the protocol complex, assigning each vertex (i.e., each local state) with a decision value.
A bit more abstractly, the data provided by the decision values can be thought of as a simplicial map $\delta : \cP \to \cO$, sending each vertex of $\cP$ to a vertex of $\cO$.
This is the key idea underlying the asynchronous computability theorem of Herlihy and Shavit~\cite{HS99}.

\begin{theorem}[Asynchronous Computability Theorem~\cite{HS99}]
 A task $\cO$ is \emph{solvable} using the protocol $\cP$ 
  if and only if there exists a simplicial map $\delta: \cP\rightarrow \cO$ that respects the task specification, that is, such that the diagram of simplicial complexes below commutes.
\begin{center}
\begin{tikzpicture}[auto]
  \node (p) at (1.5,1) {$\cP$};
  \node (o) at (3,0) {$\cO$};
  \node (i) at (0,0) {$\cI$};
  \draw[->] (i) to node {$\Xi$} (p);
  \draw[->] (p) to node {$\delta$} (o);
  \draw[->] (i) to node[swap] {$\Delta$} (o);
\end{tikzpicture}
\end{center}
\end{theorem}

This theorem is fundamental in the field of distributed computing, because it reduces a computational problem (solvability of a task) to a topological one (existence of a simplicial map).
This is extremely useful when we want to prove that no such decision map exists, because simplicial maps preserve topological invariants.

As we have seen, simplicial complexes can also be viewed as models for epistemic logic.
Thus, instead of relying on topology to find an impossibility proof, we will find an obstruction to the existence of $\delta$, based on epistemic logic.
The main tool to prove impossibility results using epistemic logic is the knowledge gain theorem~\cite{gandalf-journal}.
Intuitively, it says that the agents do not gain knowledge when they choose decision values.
Indeed, agents acquire new information during the communication phase of the protocol.
But then, the act of choosing a decision value is made locally and kept entirely private, so it cannot lead to new information being learned.

\begin{theorem}[Knowledge gain~\cite{gandalf-journal}]
Let $\delta : \cP \to \cO$ be a simplicial map, and $\phi$ a positive epistemic logic formula. Let $w \in \cP$ be a world of $\cP$.
If the formula $\phi$ is true in world $\delta(w)$, then it was already true in $w$: \,
$\cO, \delta(w) \models \phi \,\implies\, \cP, w \models \phi$.
\end{theorem}

In the above theorem, a \emph{positive} epistemic formula is a formula that only talks about what the agents know, and not about what they do not know.
This theorem, together with Herlihy and Shavit's characterization of asynchronous computability, gives a simple, concrete recipe to prove that a task is not solvable\footnote{Expert readers will notice that we are cheating a little bit here. The output complex $\cO$ is annotated with output values, rather than inputs. So we cannot evaluate epistemic formulas on $\cO$ directly. Rather, this should be done on the \emph{product update model} $\cI[\cO]$, where we view the output complex as an action model, as in~\cite{gandalf-journal}. In the remainder of the paper, as a simplification, we pretend that $\cO$ itself is about input values. All the arguments that we present carry over to $\cI[\cO]$, but require a bit more care.}:
\begin{itemize}
\item Assume for contradiction that a decision map $\delta : \cP \to \cO$ exists.
\item Choose a positive epistemic logic formula $\phi$.
\item Find a world $w \in \cP$ in the protocol complex where the formula $\phi$ is false.
\item Show that $\phi$ is true in the image $\delta(w) \in \cO$.
\item This contradicts the knowledge gain theorem, so the initial assumption was incorrect.
\end{itemize}

Such a formula $\phi$ is called a \emph{logical obstruction}.
Intuitively, the formula $\phi$ describes some amount of knowledge which is a necessary condition to be able to solve the task $\cO$, but is not achieved using protocol~$\cP$.

\subsection{In the unreliable broadcast model}
\label{sec:unreliable-algo}

In the unreliable broadcast model, majority consensus can be solved in one round, using the \emph{courteous algorithm} from~\cite{FPR2026}:

\begin{center}
\begin{minipage}{0.8\textwidth}
\hrule
\smallskip 
\textbf{Courteous algorithm:}
\smallskip 
\hrule
\smallskip 
Let $(x_a, x_b, x_c)$ denote the input values of the three agents $a$, $b$, $c$.
Each agent $i \in \Ag$ broadcasts its input value, and receives a set $v_i = \{ (j, x_j) \}$ of input values (including its own), tagged with the corresponding agent~$j$.
So the set $v_i$ contains one, two, or three values.
Agent~$i$ decides its output according to the following rules:
\begin{itemize}
\item If all values are equal, decide this value.
\item If $\card{v_i} = 2$, decide $1-x_i$.
\item If $\card{v_i} = 3$ and there is a majority of $0$'s, decide $0$.
\item If $\card{v_i} = 3$ and there is a majority of $1$'s, decide $1-x_i$.
\end{itemize}
\smallskip 
\hrule
\end{minipage}
\end{center}
        
\smallskip

One can prove operationally that this algorithm is correct, or by labeling the vertices of the protocol complex after one round, and checking that the requirements of the task are not violated.
The decision values of the algorithm can be visualized in the picture below.
Recall that the input complex of the majority consensus task for three agents is a triangulated sphere (cf.\ \cref{sec:tasks}), where each of the three agent $a$, $b$, $c$ starts with input value $0$ or $1$.
In the protocol complex, each of the $8$ input triangles of the input complex is subdivided into seven triangles, as explained in \cref{sec:distrModels}.
The picture below only shows half of the unreliable broadcast protocol complex: it shows the subdivision of four possible input configurations.
From left to right, $(a_1,b_1,c_1)$, then $(a_1, b_0, c_1)$, then $(a_0, b_0, c_1)$, and finally $(a_0, b_0, c_0)$.
So on the left of the picture, all agents have input value~$1$; and on the right, all have input value~$0$.
\begin{center}
	\vspace{-0.3cm}
	\scalebox{0.9}{
	\begin{tikzpicture}[scale=1.2, auto, font={\small}]
	\begin{pgfonlayer}{nodelayer}
		\node [style=yellowvertex] (8) at (10, 0) {$c_1$};
		\node [style=pinkvertex, label={\color{purple}$\mathbf{1}$}] (9) at (14, 0) {$a_0$};
		\node [style=yellowvertex, label={\color{teal}$\mathbf{0}$}] (10) at (16, 0) {$c_1$};
		\node [style=pinkvertex] (11) at (20, 0) {};
		\node [style=bluevertex] (12) at (12, 3.5) {$b_0$};
		\node [style=yellowvertex] (14) at (13, 5.25) {$c_1$};
		\node [style=bluevertex] (16) at (15, 8.75) {};
		\node [style=pinkvertex] (17) at (17, 5.25) {};
		\node [style=bluevertex] (18) at (18, 3.5) {};
		\node [style=bluevertex, label={\color{teal}$\mathbf{0}$}] (19) at (15, 1.75) {$b_0$};
		\node [style=yellowvertex, label={\color{teal}$\mathbf{0}$}] (20) at (16, 3.5) {$c_1$};
		\node [style=pinkvertex, label={\color{teal}$\mathbf{0}$}] (21) at (14, 3.5) {$a_0$};
		\node [style=yellowvertex, label={\color{teal}$\mathbf{0}$}] (22) at (25, 8.75) {$c_0$};
		\node [style=pinkvertex, label={\color{teal}$\mathbf{0}$}] (23) at (23, 5.25) {$a_0$};
		\node [style=yellowvertex, label={\color{teal}$\mathbf{0}$}] (24) at (22, 3.5) {$c_0$};
		\node [style=pinkvertex, label={\color{teal}$\mathbf{0}$}] (25) at (20, 0) {$a_0$};
		\node [style=bluevertex, label={\color{teal}$\mathbf{0}$}] (26) at (21, 8.75) {$b_0$};
		\node [style=yellowvertex, label={\color{teal}$\mathbf{0}$}] (27) at (19, 8.75) {$c_0$};
		\node [style=bluevertex] (28) at (15, 8.75) {$b_0$};
		\node [style=pinkvertex, label={\color{teal}$\mathbf{0}$}] (29) at (17, 5.25) {$a_0$};
		\node [style=bluevertex, label={\color{teal}$\mathbf{0}$}] (30) at (18, 3.5) {$b_0$};
		\node [style=bluevertex, label={\color{teal}$\mathbf{0}$}] (31) at (21, 5.25) {$b_0$};
		\node [style=yellowvertex, label={\color{teal}$\mathbf{0}$}] (32) at (19, 5.25) {$c_0$};
		\node [style=pinkvertex, label={\color{teal}$\mathbf{0}$}] (33) at (20, 7.25) {$a_0$};
		\node [style=bluevertex, label={\color{purple}$\mathbf{1}$}] (46) at (0, 0) {$b_1$};
		\node [style=yellowvertex, label={\color{purple}$\mathbf{1}$}] (47) at (4, 0) {$c_1$};
		\node [style=bluevertex, label={\color{purple}$\mathbf{1}$}] (48) at (6, 0) {$b_1$};
		\node [style=pinkvertex] (49) at (10, 0) {};
		\node [style=pinkvertex, label={\color{purple}$\mathbf{1}$}] (50) at (2, 3.5) {$a_1$};
		\node [style=bluevertex, label={\color{purple}$\mathbf{1}$}] (51) at (3, 5.25) {$b_1$};
		\node [style=bluevertex] (52) at (5, 8.75) {};
		\node [style=pinkvertex] (53) at (7, 5.25) {};
		\node [style=bluevertex] (54) at (8, 3.5) {};
		\node [style=pinkvertex, label={\color{purple}$\mathbf{1}$}] (55) at (5, 1.75) {$a_1$};
		\node [style=bluevertex, label={\color{purple}$\mathbf{1}$}] (56) at (6, 3.5) {$b_1$};
		\node [style=yellowvertex, label={\color{purple}$\mathbf{1}$}] (57) at (4, 3.5) {$c_1$};
		\node [style=bluevertex, label={\color{teal}$\mathbf{0}$}] (58) at (15, 8.75) {$b_0$};
		\node [style=yellowvertex, label={\color{teal}$\mathbf{0}$}] (59) at (13, 5.25) {$c_1$};
		\node [style=bluevertex, label={\color{purple}$\mathbf{1}$}] (60) at (12, 3.5) {$b_0$};
		\node [style=yellowvertex, label={\color{purple}$\mathbf{1}$}] (61) at (10, 0) {$c_1$};
		\node [style=pinkvertex, label={\color{teal}$\mathbf{0}$}] (62) at (11, 8.75) {$a_1$};
		\node [style=bluevertex, label={\color{purple}$\mathbf{1}$}] (63) at (9, 8.75) {$b_0$};
		\node [style=pinkvertex, label={\color{purple}$\mathbf{1}$}] (64) at (5, 8.75) {$a_1$};
		\node [style=yellowvertex, label={\color{purple}$\mathbf{1}$}] (65) at (7, 5.25) {$c_1$};
		\node [style=pinkvertex, label={\color{purple}$\mathbf{1}$}] (66) at (8, 3.5) {$a_1$};
		\node [style=pinkvertex, label={\color{teal}$\mathbf{0}$}] (67) at (11, 5.25) {$a_1$};
		\node [style=bluevertex, label={\color{purple}$\mathbf{1}$}] (68) at (9, 5.25) {$b_0$};
		\node [style=yellowvertex, label={\color{teal}$\mathbf{0}$}] (69) at (10, 7.25) {$c_1$};
	\end{pgfonlayer}
	\begin{pgfonlayer}{edgelayer}
		\draw [style=filled path] (8.center)
			 to (12.center)
			 to (9.center)
			 to cycle;
		\draw [style=filled path] (19.center)
			 to (10.center)
			 to (9.center)
			 to cycle;
		\draw [style=filled path] (19.center)
			 to (21.center)
			 to (20.center)
			 to cycle;
		\draw [style=filled path] (12.center)
			 to (21.center)
			 to (14.center)
			 to cycle;
		\draw [style=filled path] (17.center)
			 to (18.center)
			 to (20.center)
			 to cycle;
		\draw [style=filled path] (11.center)
			 to (10.center)
			 to (18.center)
			 to cycle;
		\draw [style=filled path] (17.center)
			 to (16.center)
			 to (14.center)
			 to cycle;
		\draw [style=filled path] (22.center)
			 to (26.center)
			 to (23.center)
			 to cycle;
		\draw [style=filled path] (31.center)
			 to (24.center)
			 to (23.center)
			 to cycle;
		\draw [style=filled path] (31.center)
			 to (33.center)
			 to (32.center)
			 to cycle;
		\draw [style=filled path] (26.center)
			 to (33.center)
			 to (27.center)
			 to cycle;
		\draw [style=filled path] (29.center)
			 to (30.center)
			 to (32.center)
			 to cycle;
		\draw [style=filled path] (25.center)
			 to (24.center)
			 to (30.center)
			 to cycle;
		\draw [style=filled path] (29.center)
			 to (28.center)
			 to (27.center)
			 to cycle;
		\draw [style=arrow] (8) to (12);
		\draw [style=arrow] (8) to (9);
		\draw [style=arrow] (9) to (19);
		\draw [style=arrow] (10) to (19);
		\draw [style=arrow] (25) to (10);
		\draw [style=arrow] (25) to (30);
		\draw [style=arrow] (25) to (24);
		\draw [style=arrow] (24) to (31);
		\draw [style=arrow] (23) to (31);
		\draw [style=arrow] (30) to (20);
		\draw [style=arrow] (29) to (20);
		\draw [style=arrow] (28) to (14);
		\draw [style=arrow] (28) to (29);
		\draw [style=arrow] (28) to (27);
		\draw [style=arrow] (27) to (33);
		\draw [style=arrow] (26) to (33);
		\draw [style=arrow] (22) to (23);
		\draw [style=arrow] (22) to (26);
		\draw [style=double arrow] (27) to (26);
		\draw [style=double arrow] (23) to (24);
		\draw [style=double arrow] (29) to (30);
		\draw [style=double arrow] (9) to (10);
		\draw [style=double arrow] (14) to (12);
		\draw [style=double arrow] (21) to (20);
		\draw [style=double arrow] (20) to (19);
		\draw [style=double arrow] (19) to (21);
		\draw [style=double arrow] (32) to (33);
		\draw [style=double arrow] (33) to (31);
		\draw [style=double arrow] (31) to (32);
		\draw [style=arrow] (29) to (32);
		\draw [style=arrow] (30) to (32);
		\draw [style=filled path] (46.center)
			 to (50.center)
			 to (47.center)
			 to cycle;
		\draw [style=filled path] (55.center)
			 to (48.center)
			 to (47.center)
			 to cycle;
		\draw [style=filled path] (55.center)
			 to (57.center)
			 to (56.center)
			 to cycle;
		\draw [style=filled path] (50.center)
			 to (57.center)
			 to (51.center)
			 to cycle;
		\draw [style=filled path] (53.center)
			 to (54.center)
			 to (56.center)
			 to cycle;
		\draw [style=filled path] (49.center)
			 to (48.center)
			 to (54.center)
			 to cycle;
		\draw [style=filled path] (53.center)
			 to (52.center)
			 to (51.center)
			 to cycle;
		\draw [style=filled path] (58.center)
			 to (62.center)
			 to (59.center)
			 to cycle;
		\draw [style=filled path] (67.center)
			 to (60.center)
			 to (59.center)
			 to cycle;
		\draw [style=filled path] (67.center)
			 to (69.center)
			 to (68.center)
			 to cycle;
		\draw [style=filled path] (62.center)
			 to (69.center)
			 to (63.center)
			 to cycle;
		\draw [style=filled path] (65.center)
			 to (66.center)
			 to (68.center)
			 to cycle;
		\draw [style=filled path] (61.center)
			 to (60.center)
			 to (66.center)
			 to cycle;
		\draw [style=filled path] (65.center)
			 to (64.center)
			 to (63.center)
			 to cycle;
		\draw [style=arrow] (46) to (50);
		\draw [style=arrow] (46) to (47);
		\draw [style=arrow] (47) to (55);
		\draw [style=arrow] (48) to (55);
		\draw [style=arrow] (61) to (48);
		\draw [style=arrow] (61) to (66);
		\draw [style=arrow] (61) to (60);
		\draw [style=arrow] (60) to (67);
		\draw [style=arrow] (59) to (67);
		\draw [style=arrow] (66) to (56);
		\draw [style=arrow] (65) to (56);
		\draw [style=arrow] (64) to (51);
		\draw [style=arrow] (64) to (65);
		\draw [style=arrow] (64) to (63);
		\draw [style=arrow] (63) to (69);
		\draw [style=arrow] (62) to (69);
		\draw [style=arrow] (58) to (59);
		\draw [style=arrow] (58) to (62);
		\draw [style=double arrow] (63) to (62);
		\draw [style=double arrow] (59) to (60);
		\draw [style=double arrow] (65) to (66);
		\draw [style=double arrow] (47) to (48);
		\draw [style=double arrow] (51) to (50);
		\draw [style=double arrow] (57) to (56);
		\draw [style=double arrow] (56) to (55);
		\draw [style=double arrow] (55) to (57);
		\draw [style=double arrow] (68) to (69);
		\draw [style=double arrow] (69) to (67);
		\draw [style=double arrow] (67) to (68);
		\draw [style=arrow] (65) to (68);
		\draw [style=arrow] (66) to (68);
		\draw [style=fill red] (50.center)
			 to (47.center)
			 to (46.center)
			 to cycle;
		\draw [style=fill red] (47.center)
			 to (55.center)
			 to (48.center)
			 to cycle;
		\draw [style=fill red] (56.center)
			 to (55.center)
			 to (57.center)
			 to cycle;
		\draw [style=fill red] (50.center)
			 to (57.center)
			 to (51.center)
			 to cycle;
		\draw [style=fill red] (51.center)
			 to (64.center)
			 to (63.center)
			 to (65.center)
			 to cycle;
		\draw [style=fill red] (65.center)
			 to (68.center)
			 to (66.center)
			 to (56.center)
			 to cycle;
		\draw [style=fill red] (61.center)
			 to (48.center)
			 to (66.center)
			 to (60.center)
			 to (9.center)
			 to cycle;
	\end{pgfonlayer}
\end{tikzpicture}
	}
	\smallskip
\end{center}
The number above each vertex ($0$ or $1$) corresponds to the decision value of the agent given this local state.
We can easily check that the algorithm is correct: for each triangle of the protocol complex, either all agents decide~$1$, or there is a majority of~$0$.

Let us analyze the knowledge of the agents in this protocol complex.
Here, we are only concerned about what the agents know about each other's input values.
In fact, all we will use is an atomic formula that says ``all agents have input value~$1$'', that we denote by $\text{all}_1$.
This atomic formula is true in the $7$ worlds with vertices $(a_1,b_1,c_1)$ on the left of the picture; and it is false everywhere else.
Moreover, let $\alpha = \{ \{a,b\}, \{a,c\}, \{b,c\}\}$, and consider the formula $\Phi = CD_\alpha (\neg \text{all}_1)$.
Namely, $\Phi$ says: ``there is common distributed knowledge among all groups of two agents that not everyone has input value~$1$''.
The red region in the picture of the protocol complex corresponds to the worlds where the formula $\Phi$ is false; and the blue region, those where $\Phi$ is true.

With that in mind, we can reformulate the majority consensus algorithm in terms of knowledge.
Indeed, notice that agents decide~$1$ exactly when they are bordering the red region; and they decide~$0$ otherwise.
So, the majority consensus algorithm presented above can be reformulated as follows: \emph{If I know the formula~$\Phi$, decide~$0$, otherwise decide~$1$}.

\subsection{In the immediate snapshot model}

The immediate snapshot model is known to be a very weak model of computation.
Starting from the input complex, running one round of the immediate snapshot protocol results in a subdivision of each simplex of the input, without introducing any holes (see \cref{sec:distrModels}).
Running several rounds of the protocol simply subdivides more and more each simplex.
Crucially, this means that the topology of the input complex is preserved, which allows to prove many impossibility results~\cite{herlihyetal:2013}.

In the case of the majority consensus task, we have seen that the binary input complex is a triangulated sphere.
When the agents communicate using the immediate snapshot protocol for any finite number of rounds, the resulting protocol complex is still a combinatorial sphere, although with exponentially many more triangles.
Importantly, after any number of rounds, the protocol complex remains $2$-connected (in the sense of edge connectivity, not homotopical connectivity, see \cref{sec:CDK}).

Consider the same formula $\Phi = CD_\alpha (\neg \text{all}_1)$ as in the previous section.
Recall that common distributed knowledge among all groups of two agents, $CD_\alpha$, amounts to checking that the formula holds everywhere in the $2$-connected component of the actual world.
Since the protocol complex $\cP$ of the immediate snapshot model is $2$-connected, the formula $\Phi$ is false everywhere in $\cP$.
Indeed, no matter which world is the actual world, there always exists a path, moving through worlds that share a common edge, reaching a world where all agents started with input value~$1$.

Let $w$ be any world of $\cP$ in which all agents started with input value~$0$, $(a_0, b_0, c_0)$.
According to the validity requirement of the majority consensus task, the only allowed decision values are $(0, 0, 0)$.
So, the decision map~$\delta$, if it exists, must send world~$w$ to the world $\delta(w) = (0, 0, 0)$ of the output complex~$\cO$.
However, in the output complex, it is easy to see that the formula $\Phi$ is true in world $\delta(w)$.
Since we have already seen that~$\Phi$ is false in~$w$, this contradicts the knowledge gain theorem. In conclusion:

\begin{theorem}
The majority consensus task is not solvable in the immediate snapshot model, for any number of rounds.
\end{theorem}

\subsection{In the test-and-set model}

We now investigate the solvability of majority consensus in the test-and-set model.

\paragraph*{In one round}
First, let us try to apply the knowledge-based algorithm from \cref{sec:unreliable-algo}: at the end of communication, if the agent knows the formula $\Phi = CD_\alpha (\neg \text{all}_1)$, it decides~$0$, otherwise it decides~$1$.
For the one round test-and-set model, the corresponding decision map can be visualized in the picture below.
As before, the red region indicates the worlds where formula~$\Phi$ is false.
Decision values are written above each vertex of the protocol complex.
\begin{center}
\scalebox{0.9}{
	\begin{tikzpicture}[scale=1.2, auto, font={\small}]
	\begin{pgfonlayer}{nodelayer}
		\node [style=yellowvertex] (8) at (0, 0) {$c_1$};
		\node [style=pinkvertex] (11) at (10, 0) {};
		\node [style=yellowvertex, label={\color{teal}$\mathbf{0}$}] (22) at (15, 8.75) {$c_0$};
		\node [style=pinkvertex, label={\color{teal}$\mathbf{0}$}] (25) at (10, 0) {$a_0$};
		\node [style=bluevertex] (28) at (5, 8.75) {$b_0$};
		\node [style=yellowvertex, label={\color{teal}$\mathbf{0}$}] (50) at (6.25, 0) {$c_1$};
		\node [style=pinkvertex, label={\color{purple}$\mathbf{1}$}] (51) at (3.75, 0) {$a_0$};
		\node [style=yellowvertex, label={\color{teal}$\mathbf{0}$}] (52) at (8.75, 8.75) {$c_0$};
		\node [style=bluevertex, label={\color{teal}$\mathbf{0}$}] (53) at (11.25, 8.75) {$b_0$};
		\node [style=yellowvertex] (54) at (3, 5.25) {$c_1$};
		\node [style=bluevertex] (55) at (2, 3.5) {$b_0$};
		\node [style=pinkvertex, label={\color{teal}$\mathbf{0}$}] (56) at (7, 5.25) {$a_0$};
		\node [style=bluevertex, label={\color{teal}$\mathbf{0}$}] (57) at (8, 3.5) {$b_0$};
		\node [style=bluevertex, label={\color{purple}$\mathbf{1}$}] (58) at (5, 2) {$b_0$};
		\node [style=pinkvertex, label={\color{purple}$\mathbf{1}$}] (59) at (4, 3.5) {$a_0$};
		\node [style=yellowvertex, label={\color{teal}$\mathbf{0}$}] (60) at (6, 3.5) {$c_1$};
		\node [style=yellowvertex, label={\color{teal}$\mathbf{0}$}] (61) at (9, 5.25) {$c_0$};
		\node [style=pinkvertex, label={\color{teal}$\mathbf{0}$}] (62) at (10, 7) {$a_0$};
		\node [style=bluevertex, label={\color{teal}$\mathbf{0}$}] (63) at (11, 5.25) {$b_0$};
		\node [style=yellowvertex, label={\color{teal}$\mathbf{0}$}] (64) at (12, 3.5) {$c_0$};
		\node [style=pinkvertex, label={\color{teal}$\mathbf{0}$}] (65) at (13, 5.25) {$a_0$};
		\node [style=bluevertex, label={\color{purple}$\mathbf{1}$}] (76) at (-10, 0) {$b_1$};
		\node [style=pinkvertex] (77) at (0, 0) {};
		\node [style=bluevertex, label={\color{teal}$\mathbf{0}$}] (78) at (5, 8.75) {$b_0$};
		\node [style=yellowvertex, label={\color{purple}$\mathbf{1}$}] (79) at (0, 0) {$c_1$};
		\node [style=pinkvertex, label={\color{purple}$\mathbf{1}$}] (80) at (-5, 8.75) {$a_1$};
		\node [style=bluevertex, label={\color{purple}$\mathbf{1}$}] (81) at (-3.75, 0) {$b_1$};
		\node [style=yellowvertex, label={\color{purple}$\mathbf{1}$}] (82) at (-6.25, 0) {$c_1$};
		\node [style=bluevertex, label={\color{purple}$\mathbf{1}$}] (83) at (-1.25, 8.75) {$b_0$};
		\node [style=pinkvertex, label={\color{teal}$\mathbf{0}$}] (84) at (1.25, 8.75) {$a_1$};
		\node [style=bluevertex, label={\color{purple}$\mathbf{1}$}] (85) at (-7, 5.25) {$b_1$};
		\node [style=pinkvertex, label={\color{purple}$\mathbf{1}$}] (86) at (-8, 3.5) {$a_1$};
		\node [style=yellowvertex, label={\color{purple}$\mathbf{1}$}] (87) at (-3, 5.25) {$c_1$};
		\node [style=pinkvertex, label={\color{purple}$\mathbf{1}$}] (88) at (-2, 3.5) {$a_1$};
		\node [style=pinkvertex, label={\color{purple}$\mathbf{1}$}] (89) at (-5, 2) {$a_1$};
		\node [style=yellowvertex, label={\color{purple}$\mathbf{1}$}] (90) at (-6, 3.5) {$c_1$};
		\node [style=bluevertex, label={\color{purple}$\mathbf{1}$}] (91) at (-4, 3.5) {$b_1$};
		\node [style=bluevertex, label={\color{purple}$\mathbf{1}$}] (92) at (-1, 5.25) {$b_0$};
		\node [style=yellowvertex, label={\color{purple}$\mathbf{1}$}] (93) at (0, 7) {$c_1$};
		\node [style=pinkvertex, label={\color{purple}$\mathbf{1}$}] (94) at (1, 5.25) {$a_1$};
		\node [style=bluevertex, label={\color{purple}$\mathbf{1}$}] (95) at (2, 3.5) {$b_0$};
		\node [style=yellowvertex, label={\color{teal}$\mathbf{0}$}] (96) at (3, 5.25) {$c_1$};
		\node [style=none] (97) at (0, 7) {};
		\node [style=none] (98) at (5, 8.75) {};
		\node [style=none] (99) at (1, 5.25) {};
	\end{pgfonlayer}
	\begin{pgfonlayer}{edgelayer}
		\draw [style=arrow] (8) to (55);
		\draw [style=arrow] (8) to (51);
		\draw [style=arrow] (8) to (59);
		\draw [style=arrow] (8) to (58);
		\draw [style=arrow] (25) to (50);
		\draw [style=arrow] (25) to (57);
		\draw [style=arrow] (25) to (58);
		\draw [style=arrow] (25) to (60);
		\draw [style=arrow] (25) to (61);
		\draw [style=arrow] (25) to (63);
		\draw [style=arrow] (25) to (64);
		\draw [style=arrow] (22) to (65);
		\draw [style=arrow] (22) to (63);
		\draw [style=arrow] (22) to (53);
		\draw [style=arrow] (22) to (62);
		\draw [style=arrow] (28) to (52);
		\draw [style=arrow] (28) to (62);
		\draw [style=arrow] (28) to (61);
		\draw [style=arrow] (28) to (56);
		\draw [style=arrow] (28) to (60);
		\draw [style=arrow] (28) to (59);
		\draw [style=arrow] (54) to (59);
		\draw [style=arrow] (55) to (59);
		\draw [style=arrow] (51) to (58);
		\draw [style=arrow] (50) to (58);
		\draw [style=arrow] (56) to (60);
		\draw [style=arrow] (57) to (60);
		\draw [style=arrow] (57) to (61);
		\draw [style=arrow] (56) to (61);
		\draw [style=arrow] (52) to (62);
		\draw [style=arrow] (53) to (62);
		\draw [style=arrow] (65) to (63);
		\draw [style=arrow] (64) to (63);
		\draw [style=double arrow] (59) to (60);
		\draw [style=double arrow] (60) to (58);
		\draw [style=double arrow] (58) to (59);
		\draw [style=double arrow] (62) to (61);
		\draw [style=double arrow] (61) to (63);
		\draw [style=double arrow] (63) to (62);
		\draw [style=arrow] (28) to (54);
		\draw [style=filled path] (58.center)
			 to (51.center)
			 to (8.center)
			 to (55.center)
			 to (59.center)
			 to cycle;
		\draw [style=filled path] (25.center)
			 to (64.center)
			 to (63.center)
			 to (61.center)
			 to (57.center)
			 to (60.center)
			 to (58.center)
			 to (50.center)
			 to cycle;
		\draw [style=filled path] (60.center)
			 to (56.center)
			 to (61.center)
			 to (62.center)
			 to (52.center)
			 to (28.center)
			 to (54.center)
			 to (59.center)
			 to cycle;
		\draw [style=filled path] (53.center)
			 to (62.center)
			 to (63.center)
			 to (65.center)
			 to (22.center)
			 to cycle;
		\draw [style=arrow] (76) to (86);
		\draw [style=arrow] (76) to (82);
		\draw [style=arrow] (76) to (90);
		\draw [style=arrow] (76) to (89);
		\draw [style=arrow] (79) to (81);
		\draw [style=arrow] (79) to (88);
		\draw [style=arrow] (79) to (89);
		\draw [style=arrow] (79) to (91);
		\draw [style=arrow] (79) to (92);
		\draw [style=arrow] (79) to (94);
		\draw [style=arrow] (79) to (95);
		\draw [style=arrow] (78) to (96);
		\draw [style=arrow] (78) to (94);
		\draw [style=arrow] (78) to (84);
		\draw [style=arrow] (78) to (93);
		\draw [style=arrow] (80) to (83);
		\draw [style=arrow] (80) to (93);
		\draw [style=arrow] (80) to (92);
		\draw [style=arrow] (80) to (87);
		\draw [style=arrow] (80) to (91);
		\draw [style=arrow] (80) to (90);
		\draw [style=arrow] (85) to (90);
		\draw [style=arrow] (86) to (90);
		\draw [style=arrow] (82) to (89);
		\draw [style=arrow] (81) to (89);
		\draw [style=arrow] (87) to (91);
		\draw [style=arrow] (88) to (91);
		\draw [style=arrow] (88) to (92);
		\draw [style=arrow] (87) to (92);
		\draw [style=arrow] (83) to (93);
		\draw [style=arrow] (84) to (93);
		\draw [style=arrow] (96) to (94);
		\draw [style=arrow] (95) to (94);
		\draw [style=double arrow] (90) to (91);
		\draw [style=double arrow] (91) to (89);
		\draw [style=double arrow] (89) to (90);
		\draw [style=double arrow] (93) to (92);
		\draw [style=double arrow] (92) to (94);
		\draw [style=double arrow] (94) to (93);
		\draw [style=arrow] (80) to (85);
		\draw [style=filled path] (89.center)
			 to (82.center)
			 to (76.center)
			 to (86.center)
			 to (90.center)
			 to cycle;
		\draw [style=filled path] (79.center)
			 to (95.center)
			 to (94.center)
			 to (92.center)
			 to (88.center)
			 to (91.center)
			 to (89.center)
			 to (81.center)
			 to cycle;
		\draw [style=filled path] (91.center)
			 to (87.center)
			 to (92.center)
			 to (93.center)
			 to (83.center)
			 to (80.center)
			 to (85.center)
			 to (90.center)
			 to cycle;
		\draw [style=filled path] (84.center)
			 to (93.center)
			 to (94.center)
			 to (96.center)
			 to (78.center)
			 to cycle;
		\draw [style=fill red] (82.center)
			 to (76.center)
			 to (86.center)
			 to (90.center)
			 to (89.center)
			 to cycle;
		\draw [style=fill red] (91.center)
			 to (89.center)
			 to (81.center)
			 to (79.center)
			 to (51.center)
			 to (58.center)
			 to (59.center)
			 to (95.center)
			 to (94.center)
			 to (92.center)
			 to (88.center)
			 to cycle;
		\draw [style=fill red] (87.center)
			 to (92.center)
			 to (93.center)
			 to (83.center)
			 to (80.center)
			 to (85.center)
			 to (90.center)
			 to (91.center)
			 to cycle;
		\draw [style=crossed section] (98.center)
			 to (97.center)
			 to (99.center)
			 to cycle;
	\end{pgfonlayer}
\end{tikzpicture}
}
\end{center}
Notice that the resulting algorithm is incorrect: in the crosshatched triangle, with input values $(a_1, b_0, c_1)$, the three decision values are $(1, 0, 1)$. This contradict the majority agreement requirement of the task, since there is a majority of~$1$'s.
The issue is that, in this world, only one agent knows the formula~$\Phi$: both $a$ and $c$ are on the border of the red region, so they consider possible that $\Phi$ might be false.

However, in the majority consensus output complex~$\cO$, this situation cannot happen: in all worlds, either $0$ or $\geq 2$ agents know $\Phi$.
This statement\footnote{Not exactly: ``$0$ agents know $\Phi$'' is not a positive formula, so we need to reformulate a bit.} will serve as our obstruction formula to show that majority consensus is not solvable in the one round test-and-set model.

\begin{theorem}
Majority consensus is not solvable in the one round test-and-set model.
\end{theorem}
\begin{proof}[Proof (sketch)]
Assume there exists a decision map $\delta : \cP \to \cO$ satisfying the task specification.
Let $\alpha = \{ \{a,b\}, \{a,c\}, \{b,c\}\}$ contain all groups of two agents, and define the two formulas $\Phi_1 = CD_\alpha (\neg \text{all}_1)$ and $\Phi_0 = CD_\alpha (\neg \text{all}_0)$.
Then, we define the following obstruction formula, intuitively, ``either $\Phi_0$ is true, or at least two agents know $\Phi_1$''.
\[\Psi \;=\; \Phi_0 \lor (K_a\,\Phi_1 \land K_b\,\Phi_1) \lor (K_a\,\Phi_1 \land K_c\,\Phi_1) \lor (K_b\,\Phi_1 \land K_c\,\Phi_1)\]
The formula $\Psi$ is false in the crosshatched world~$w$ of the protocol complex~$\cP$.
However,~$\Psi$ is true everywhere in the output complex~$\cO$ of majority consensus.
So, no matter where the simplicial map $\delta$ sends the world $w$, we always have that~$\Psi$ is true in $\delta(w)$, but false in~$w$.
This contradicts the knowledge gain theorem, so the decision map $\delta$ cannot exist.
\end{proof}

\paragraph*{In two rounds}

In the previous one-round attempt, there was only one world (in fact three, if we draw the full protocol complex) where the specification of the task is not met.
It seems that all we need to do is to separate the two faulty vertices $a_1$ and $c_1$ to make the algorithm work.
Indeed, one extra round of test-and-set suffices to solve majority consensus.
In fact, the full second round is not needed: all we need to do is an extra write/read, between the two agents that lost the first round.

\begin{center}
\begin{minipage}{0.8\textwidth}
\hrule
\smallskip 
\textbf{Test-and-set algorithm to solve majority consensus:}
\smallskip 
\hrule
\vspace{-4pt}
\begin{itemize}
\item First round: all agents write, call Test\&Set, and then read.
\item Agents decide as in the one-round case, except in the following case.
\item If agent~$i$ started with value~$1$, loses the Test\&Set, and reads values $0$ and $1$ from the other two agents, then do an extra write and read.
\item If agent~$i$ sees another value during the second round, decide $0$, otherwise decide $1$.
\end{itemize}
\smallskip 
\hrule
\end{minipage}
\end{center}
\smallskip

The protocol complex arising from this algorithm is almost the same as the one of the one-round case, except that the crosshatched triangle is further subdivided as follows:
\begin{center}
\scalebox{0.8}{
	\begin{tikzpicture}[scale=1, auto, font={\small}]
	\begin{pgfonlayer}{nodelayer}
		\node [style=bluevertex, label={\color{teal}$\mathbf{0}$}] (78) at (15, 9.75) {$b_0$};
		\node [style=yellowvertex, label={\color{purple}$\mathbf{1}$}] (93) at (0, 7) {$c_1$};
		\node [style=pinkvertex, label={\color{teal}$\mathbf{0}$}] (94) at (1, 5.25) {$a_1$};
		\node [style=yellowvertex, label={\color{teal}$\mathbf{0}$}] (95) at (2, 3.5) {$c_1$};
		\node [style=pinkvertex, label={\color{purple}$\mathbf{1}$}] (96) at (3, 1.75) {$a_1$};
	\end{pgfonlayer}
	\begin{pgfonlayer}{edgelayer}
		\draw [style=filled path] (93.center)
			 to (78.center)
			 to (96.center)
			 to (95.center)
			 to (94.center)
			 to cycle;
		\draw [style=filled path] (94) to (78);
		\draw [style=filled path] (78) to (95);
		\draw [style=double arrow] (94) to (95);
		\draw [style=arrow] (93) to (94);
		\draw [style=arrow] (96) to (95);
	\end{pgfonlayer}
\end{tikzpicture}
}
\vspace{-0.5cm}
\end{center}
The arrows indicate which values were exchanged during the second write/read between the two agents $a$ and $c$.
As we can see, all three triangles satisfy the task specification, so the algorithm is correct.

\section{Conclusion}
\label{sec:concl}

We have presented an introduction to simplicial models for an economics audience, somewhat informal, but self-contained and avoiding technical details, exposing the  interactions among knowledge, distributed computing,  and combinatorial topology.
We focused on the following notions, which deserve to be better known by this audience: simplicial models, distributed knowledge, models of communication, and tasks specifications.

First, while Kripke frame semantics is ubiquitous in the study of multi-agent systems, we described a dual semantics based on simplicial complexes, that exposes a topological structure.  We illustrated the value of this to reason about distributed knowledge, and common distributed knowledge. But there is more sophisticated topological information that can be exploited, related to the homology of the underlying simplicial complex, directly using combinatorial topology as in~\cite{herlihyetal:2013} and epistemically, as  in~\cite{Nishimura24}.

There is plenty of previous work on dynamics, to study how the knowledge of agents evolves after communicating with each other, based on Kripke models; especially close to our approach is Dynamic Epistemic Logic~\cite{sep-dynamic-epistemic,DEL}. We used communication patterns to show how a simplicial model evolves with communication.
A main goal was to show how a distributed model of communication can be defined, and how different models allow for different knowledge gain by the agents.
We presented three examples, immediate snapshots, broadcast and test-and-set.

Finally, we showed that simplicial models can also be used as a specification language, for distributed tasks,
modeling the knowledge that the agents should gain.
We stress that the task specification is independent of the communication model~\cite{Rajsbaum22tasks}.
We considered a weaker version of consensus, majority consensus, and showed how in some models the agents can indeed gain enough knowledge to solve the task, while in some other not.
These examples rely on the notion of common distributed knowledge.

There are  numerous possible extensions to this work.
Iterated knowledge has been considered in the past, and its relation to solving approximate agreement (where agents agree on values close to each other)~\cite{Armenta-SeguraL22,gandalf-journal}. We would like to do an analogous analysis for iterated distributed knowledge, and the generalization of majority consensus to approximate majority consensus~\cite{AttiyaFPR25}.
Clearly, generalizing our examples from three agents to any number of agents would be of interest.
Also, exploring the distributed knowledge technique more formally and generally, describing the tasks for which it can be applied. In particular, it would be interesting to apply it in setting where agents may fail, both by crashing or by arbitrary Byzantine failures.

\bibliography{bibliography.bib}

%
%

%

\end{document}